\documentclass[conference]{IEEEtran}
\IEEEoverridecommandlockouts
\usepackage{cite}
\usepackage{amsmath,amssymb,amsfonts}
\usepackage{graphicx}
\usepackage{textcomp}
\usepackage{xcolor}
\usepackage{amsmath}
\usepackage{graphics}
\usepackage{subcaption}
\usepackage{float}
\usepackage{enumerate}
\usepackage{balance}
\usepackage{todonotes}
\usepackage{mathtools}
\usepackage{booktabs}
\usepackage[linesnumbered,ruled,vlined]{algorithm2e}
\usepackage[a-1b]{pdfx}
\usepackage[nameinlink]{cleveref}
\usepackage{lipsum}
\usepackage{titlesec}
\usepackage{amsthm}
\usepackage{amsmath}
\usepackage{enumitem}
\usepackage{xspace}
\usepackage{setspace}
\usepackage{thmtools, thm-restate}
\newtheorem{example}{Example}
\newtheorem{definition}{Definition}
\newtheorem{lemma}{Lemma}
\newtheorem{theorem}{Theorem}

\hypersetup{
    colorlinks=false,
    pdfborder={0 0 0},
}

\newcommand{\SubHyMa}{SubHyMa\xspace}
\newcommand{\HyperISO}{HyperISO\xspace}
\newcommand{\HGMatch}{HGMatch\xspace}
\newcommand{\GuP}{GuP\xspace}
\newcommand{\OHMiner}{OHMiner\xspace}
\newcommand{\MaCH}{MaCH\xspace}

\newcommand{\abs}[1]{{|{#1}|}}
\newcommand{\I}{\mathcal{I}}
\newcommand\availabilityurl{https://github.com/SNUCSE-CTA/MaCH}
\everypar{\looseness=-1}

\renewcommand{\paragraph}[1]{\noindent\textbf{#1.}}

\def\BibTeX{{\rm B\kern-.05em{\sc i\kern-.025em b}\kern-.08em
    T\kern-.1667em\lower.7ex\hbox{E}\kern-.125emX}}

\begin{document}

\newif\ifsubmit

\ifsubmit
\title{Efficient Hypergraph Pattern Matching via Match-and-Filter and Intersection Constraint}
\else
\title{Efficient Hypergraph Pattern Matching via Match-and-Filter and Intersection Constraint}
\fi

\author
{
    \IEEEauthorblockN{Siwoo Song$^{\dagger}$, Wonseok Shin$^{\ddagger}$, Kunsoo Park$^{*\mathsection}$, Giuseppe F. Italiano$^{\mathparagraph}$, Zhengyi Yang$^{\lozenge}$, Wenjie Zhang$^{\lozenge}$}
    \IEEEauthorblockA{
    $^{\dagger}$\textit{Samsung Research, Korea} \hspace{3mm} $^{\ddagger}$\textit{Standigm Inc, Korea} \hspace{3mm}
    $^{\mathsection}$\textit{Seoul National University, Korea} \\ 
    $^{\mathparagraph}$\textit{LUISS University, Italy} \hspace{3mm} $^{\lozenge}$\textit{University of New South Wales, Australia}\\
     \texttt{$^{\dagger}$s6uos.song@samsung.com \hspace{3mm} $^{\ddagger}$wonseok.shin@standigm.com\hspace{3mm}
     $^{\mathsection}$kpark@theory.snu.ac.kr}\\
    \texttt{$^{\mathparagraph}$gitaliano@luiss.it \hspace{3mm} $^{\lozenge}$\{zhengyi.yang,wenjie.zhang\}@unsw.edu.au} 
    }
}

\maketitle

\begin{abstract}
A hypergraph is a generalization of a graph, in which a hyperedge can connect multiple vertices, modeling complex relationships involving multiple vertices simultaneously.
Hypergraph pattern matching, which is to find all isomorphic embeddings of a query hypergraph in a data hypergraph, is one of the fundamental problems.
In this paper, we present a novel algorithm for hypergraph pattern matching by introducing (1) the intersection constraint, a necessary and sufficient condition for valid embeddings, which significantly speeds up the verification process,
(2) the candidate hyperedge space, a data structure that stores potential mappings between hyperedges in the query hypergraph and the data hypergraph,
and (3) the Match-and-Filter framework, which interleaves matching and filtering operations to maintain only compatible candidates in the candidate hyperedge space during backtracking.
Experimental results on real-world datasets demonstrate that our algorithm significantly outperforms the state-of-the-art algorithms, by up to orders of magnitude in terms of query processing time.
\end{abstract}


\section{Introduction}

\begin{figure}[t]
    \centering
    \begin{subfigure}{0.38\linewidth}
        \centering
        \includegraphics[trim={0 0 0 0.5cm},clip,width=\linewidth]{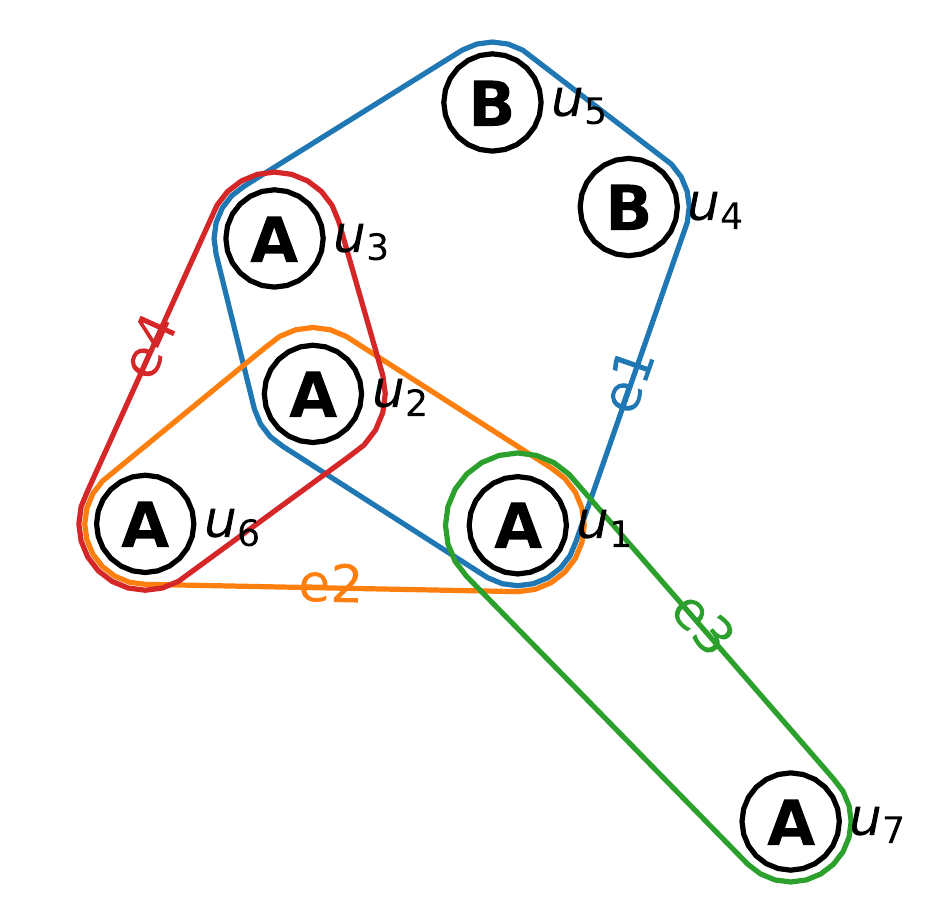}
        \caption{Query hypergraph $q$}
        \label{subfig:query}
    \end{subfigure}
    \begin{subfigure}{0.45\linewidth}
        \centering
        \includegraphics[trim={0 0 0 0.5cm},clip,width=\linewidth]{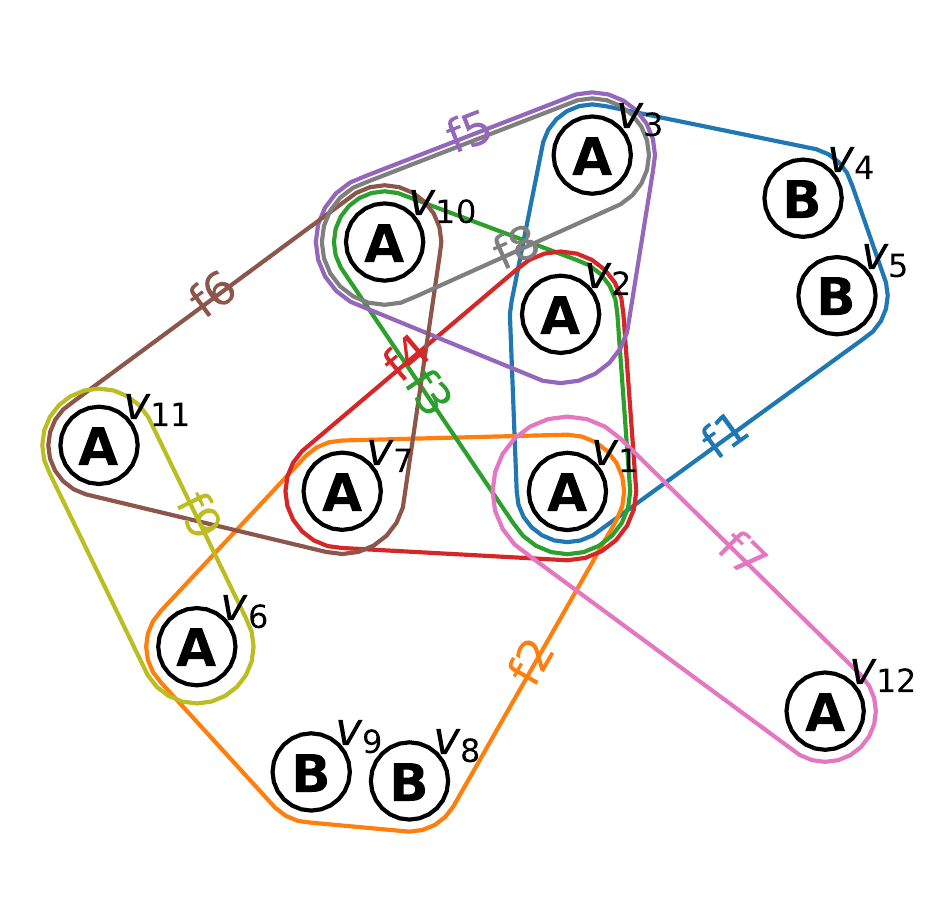}
        \caption{Data hypergraph $H$}
        \label{subfig:data}
    \end{subfigure}

    \begin{subfigure}{0.4\linewidth}
        \centering
        \includegraphics[width=0.9\linewidth]{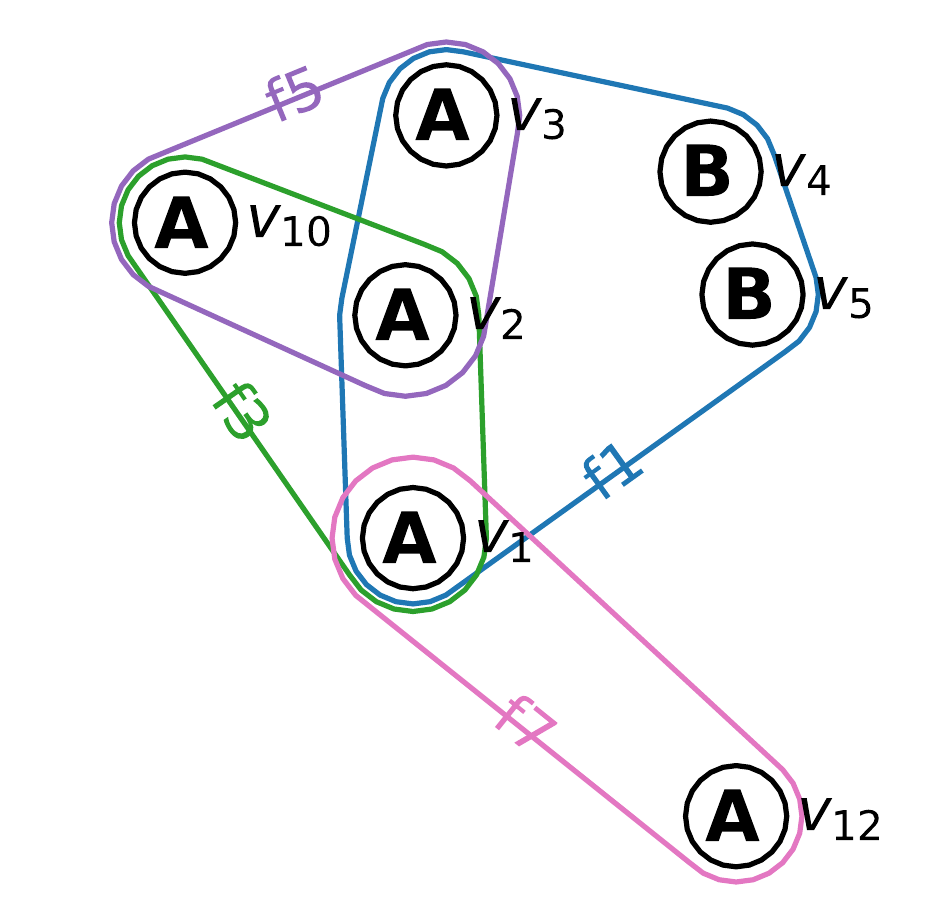}
        \caption{Embedding of $q$ in $H$}
        \label{subfig:embedding}
    \end{subfigure}
    \begin{subfigure}{0.42\linewidth}
        \centering
        \includegraphics[width=0.9\linewidth]{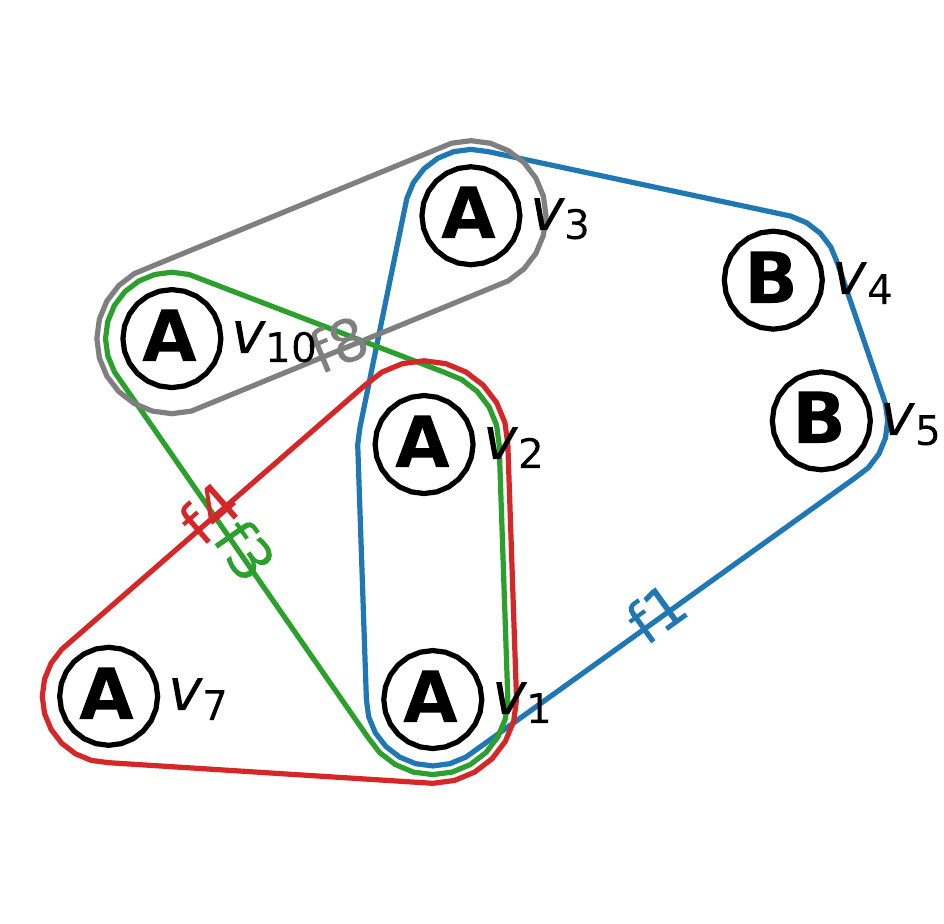}
        \caption{Wrong mapping}
        \label{subfig:wrongmapping}
    \end{subfigure}
    \caption{Example of hypergraph pattern matching where `A' and `B' are vertex labels, $u_i$ and $v_i$ are vertex IDs, and $e_i$ and $f_i$ are hyperedge IDs.}
    \label{fig:hypergraphs}
\end{figure}

Graphs are widely used to model relationships in various domains, such as social networks, bioinformatics, and chemistry. Traditional graphs represent pairwise relationships between vertices, in which an edge connects exactly two vertices. However, this pairwise representation has limitations when it comes to capturing complex relationships involving multiple vertices simultaneously that are common in many real-world scenarios \cite{ApplicationPPI, ApplicationPPI2, ApplicationSN, ApplicationBIO}. A hypergraph is a generalization of a graph by allowing edges to connect any number of vertices.
In recent years, hypergraphs have gained increasing attention as they can better represent complex relationships that arise in diverse domains, including chemical reaction networks \cite{CRN}, protein interactions \cite{ApplicationPPI, ApplicationPPI2}, knowledge bases \cite{fatemi2021knowledge}, collaboration networks \cite{lung2018hypergraph}, social networks \cite{ApplicationSN}, and electronic circuits \cite{Circuit0, Circuit1}.

Pattern matching, which is to find all isomorphic embeddings of a query object in a larger object, is one of the fundamental problems in computer science, and it has been studied in various objects such as strings, trees, graphs, and hypergraphs.
If the objects are strings, the pattern matching problem is called string matching \cite{KMP, BM}; if trees, tree pattern matching \cite{TreePatternMatching, PolyTreePatternMatching}; if graphs, subgraph matching \cite{DAF, CFLMatch}; if hypergraphs, hypergraph pattern matching (also known as subhypergraph matching \cite{HGMatch}). In this paper we focus on hypergraph pattern matching, which is a fundamental problem in understanding and analyzing hypergraph data.
In \Cref{fig:hypergraphs}, for example, the subhypergraph in \Cref{subfig:embedding} is an embedding (i.e., mapping $\{(e_1,f_1), (e_2,f_3),(e_3,f_7), (e_4,f_5)\}$) of the query hypergraph $q$ (\Cref{subfig:query}) in the data hypergraph $H$ (\Cref{subfig:data}).

Hypergraph pattern matching is essential in various applications. 
For example, electronic circuits are naturally represented as hypergraphs (\Cref{fig:Circuit}), where vertices model devices (such as transistors) and a hyperedge models a net connecting multiple devices simultaneously. 
By hypergraph pattern matching on electronic circuits, we can verify whether patented subcircuits appear in a larger circuit or locate problematic subcircuit patterns that may cause failures. 
However, hypergraph pattern matching is an NP-hard problem \cite{GareyAndJohnson}, requiring substantial computational time for large hypergraph data or queries.

\begin{figure}[t]
    \begin{subfigure}{0.49\linewidth}
        \centering
        \includegraphics[width=\linewidth,trim= 0 20 0 10,clip]{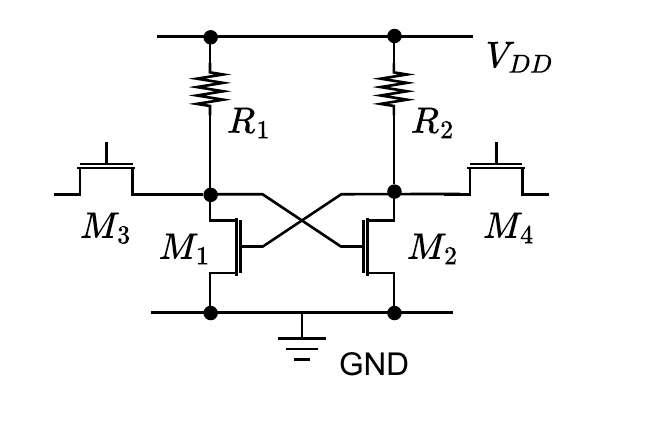}
        \caption{Transistor-level Circuit}
        \label{subfig:AnalogCircuit}
    \end{subfigure}
    \begin{subfigure}{0.49\linewidth}
\includegraphics[width=\linewidth]{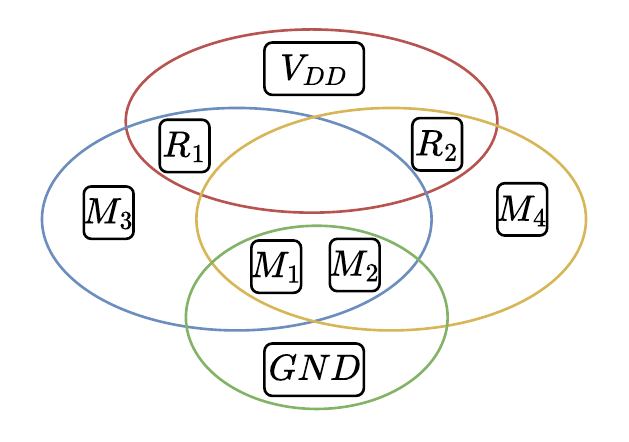}
        \caption{Hypergraph}
        \label{subfig:HypergraphCircuit}
    \end{subfigure}
    \caption{Hypergraph representation of a transistor-level circuit, where vertices are devices (transistors $M_1$ to $M_4$, resistors $R_1$ and $R_2$, $V_{DD}$, and $GND$) and a hyperedges is a net connecting multiple devices simultaneously.}
    \label{fig:Circuit}
    \vspace{0.5em}
\end{figure}

\paragraph{Existing Approaches and Limitations}
There are several approaches for hypergraph pattern matching. One approach is to transform a hypergraph into a bipartite graph. In this method, vertices and hyperedges from the original hypergraph become vertices in the bipartite graph, with edges representing incidence relationships.
Although this approach allows us to apply existing subgraph matching algorithms, it fails to utilize the properties specific to hypergraphs, resulting in limited performances \cite{SubHyMa}.

To address such issues, recent works \cite{IHSfilter, HyperIso, SubHyMa} extend an algorithm for subgraph matching to hypergraph pattern matching instead of transforming the hypergraph. 
Ha et al. \cite{IHSfilter} extend the subgraph matching algorithm Turbo\textsubscript{ISO} \cite{TurboIso}, and Yang et al. \cite{HGMatch} extend more recent algorithms, namely CFL \cite{CFLMatch}, DAF \cite{DAF}, and CECI~\cite{CECI}, to hypergraph pattern matching.

\HGMatch \cite{HGMatch} introduces the Match-by-Hyperedge framework, which maps query hyperedges directly to data hyperedges instead of mapping query vertices. Additionally, \HGMatch proposes a necessary and sufficient condition for checking the validity of embeddings, which doesn't require computing a mapping of vertices. These techniques avoid redundant computation in mapping vertices and significantly improves efficiency.
However, its methods for generating hyperedge candidates on-the-fly during the matching phase and verifying these candidates through repeated checks of hyperedges consume significant processing time.

%Very recently, Qi et al. \cite{OHMiner} suggested a verification technique based on overlap intersection graph to solve these inefficiencies.
Very recently, \OHMiner \cite{OHMiner}, the state-of-the-art hypergraph pattern matching algorithm, suggested a verification technique based on an overlap intersection graph to solve these inefficiencies.
However, their approach still requires set intersection operations, which limits overall performance. Furthermore, its high memory overhead makes it not viable for handling large-scale hypergraphs.

\paragraph{Contributions}
Compared to well-studied subgraph matching, hypergraph pattern matching presents unique challenges due to the complex connectivity patterns in hypergraphs. This complexity significantly increases the difficulty of finding valid embeddings that satisfy all constraints imposed by a query hypergraph.

However, these complex connectivity patterns in hypergraphs also provide opportunities for effective pruning of the search space by enabling early detection of invalid embeddings. 
Thus, strong yet effective constraints for hypergraph pattern matching are crucial for developing efficient algorithms. These constraints can not only prune the search space initially but also continue to prune it as the matching progresses, leading to a more efficient overall matching process.

In this paper, we present a new algorithm, \MaCH (\underline{Ma}t\underline{c}h-and-Filter for \underline{H}ypergraph Pattern Matching).
\begin{enumerate}[leftmargin=*]
    \item We introduce two novel constraints, the \emph{connectivity constraint} for the relationship between two hyperedges (e.g., two query hyperedges $e_1$ and $e_2$ share two vertices both labeled `A') and the \emph{intersection constraint} for the relationships involving three or more hyperedges (e.g., $e_1$, $e_2$, and $e_3$ share a vertex labeled `A'). These constraints effectively capture the complex connectivity patterns in hypergraphs, enabling efficient hypergraph pattern matching. Especially, we prove that the intersection constraint is a necessary and sufficient condition for valid embeddings (\Cref{thm:IntersectionSignature} and \Cref{thm:HILCequivalence}). While \HGMatch also proposed a necessary and sufficient condition for valid embeddings, called the equivalence of vertex profiles, our condition is faster both theoretically and practically.

    \item We build a new auxiliary data structure, \emph{candidate hyperedge space} (CHS), which stores candidates for query hyperedges rather than query vertices. CHS better handles vertex automorphisms and the relationships among multiple vertices in hypergraphs, overcoming limitations of vertex-based structures.

    \HGMatch enumerates embeddings on-the-fly, which may lead to a potentially large search space.
    CHS allows us to filter out unpromising candidates more effectively.
    We apply the connectivity constraint to filter out invalid candidates (which cannot be included in any embeddings) in the CHS, significantly reducing the search space.
    Our experiment shows that up to 99.3\% of candidates are removed by the connectivity constraint on the CHS.
    
    \item We propose a novel framework, \emph{Match-and-Filter}, which interleaves matching and filtering operations during backtracking. 
 Subgraph matching algorithms typically follow a filtering and matching framework, in which all filtering occurs before the matching process begins. In contrast, our framework integrates filtering operations directly into the matching process.
This framework is particularly well-suited for hypergraph pattern matching in which a hyperedge can connect multiple vertices simultaneously, leading to complex connectivity patterns that are hard to capture in a filtering phase.

We use both the intersection constraint and the connectivity constraint to filter out invalid candidates in the CHS as the matching progresses. This match-and-filter significantly prunes the search space, leading to a speedup of up to 20 times in query processing time, when compared to the matching without filtering.
\end{enumerate}

Experiments on real-world datasets demonstrate that \MaCH significantly outperforms the state-of-the-art hypergraph pattern matching algorithms, \HGMatch and \OHMiner, by up to orders of magnitude in terms of query processing time.

\ifsubmit
The source code and Appendix, containing omitted proofs and experiments, are available online.\footnote{\availabilityurl}
\else
The source code is available online.\footnote{\availabilityurl}
\fi

\section{Preliminaries}
\label{sec:Preliminaries}

\subsection{Problem Statement}

In this paper, we focus on simple, undirected, and connected hypergraphs with vertices labeled, while the techniques we propose can be readily extended to non-simple, disconnected, hyperedge-labeled, or unlabeled hypergraphs.

\begin{definition}[Hypergraph]
    A hypergraph $H$ is defined as a tuple $H=(V,E,L)$ where $V$ is a finite set of vertices and $E\subseteq 2^V$ is a set of non-empty subsets of $V$ called hyperedges, which satisfies $ \bigcup_{e\in E} e = V$. $L$ is a label function that assigns each vertex $v$ a label in $\Sigma$, which is the set of labels.
\end{definition}

In a hypergraph $H = (V_H, E_H, L_H)$, adjacency and incidence are defined as follows. Two vertices $u, v \in V_H$ are adjacent if there exists a hyperedge $e \in E_H$ such that $\{u, v\} \subseteq e$. Two hyperedges $e, f \in E_H$ are adjacent if $e \cap f \neq \emptyset$.
A vertex $v \in V_H$ and a hyperedge $e \in E_H$ are incident if $v \in e$.
The arity of a hyperedge $e$, denoted by $\abs{e}$, is the number of vertices in $e$. The average arity $\overline{a_H}$ of $H$ is defined as $\frac{\sum_{e\in E_H} \abs{e}}{\abs{E_H}}$, and the maximum arity $a_H^{\max}$ is defined as $\max_{e\in E_H} \abs{e}$.

A hypergraph $H'=(V_{H'}, E_{H'}, L_H)$ is a subhypergraph of $H$ if $V_{H'}\subseteq V_{H}$ and $E_{H'}\subseteq E_{H}$.

\begin{definition}[Subhypergraph isomorphism]
    Given a query hypergraph $q=(V_q, E_q, L_q)$ and a data hypergraph $H=(V_H, E_H, L_H)$, $q$ is subhypergraph isomorphic to $H$ if and only if there is an injective mapping $\phi: V_q\rightarrow V_H$ such that, $\forall u\in V_q$, $L_q(u)=L_H(\phi(u))$ and $\forall e_q=\{u_1,\dots, u_k\}\in E_q$, $\exists e_H = \{\phi(u_1),\dots, \phi(u_k)\}\in E_H$. We call such an $\phi$ as subhypergraph isomorphism.
\end{definition}

For $e=\{u_1,\dots, u_k\}$ and subhypergraph isomorphism $\phi$, we use the notion $\phi(e) = \{\phi(u_1),\dots, \phi(u_k)\}$ to denote the mapped hyperedge in the data hypergraph.
A \emph{subhypergraph embedding} is represented as the set of hyperedge pairs $\{(e_1, \phi(e_1)), (e_2, \phi(e_2)),\allowbreak\dots,\allowbreak (e_m, \phi(e_m))\}$ for the set of query hyperedges $E_q=\{e_1,e_2,\dots,e_m\}$. As in \HGMatch \cite{HGMatch}, embeddings are considered equivalent if their sets of hyperedge pairs are identical, even if the underlying vertex mappings differ.
A subhypergraph of the query hypergraph $q$ is called a partial query. An embedding of a partial query in $H$ is called a partial embedding.

In \Cref{fig:hypergraphs}, $\{(e_1, f_1),(e_2,f_3),(e_3,f_7),(e_4,f_5)\}$ is an embedding of $q$ in $H$. One possible underlying subhypergraph isomorphism is $\{(u_1,v_1),\allowbreak (u_2,v_2),\allowbreak (u_3,v_3),(u_4,v_4),(u_5, v_5), (u_6,v_{10}),(u_7,v_{12})\}$. We can obtain another valid subhypergraph isomorphism by swapping the mappings of $u_4$ and $u_5$, resulting in $(u_4, v_5)$ and $(u_5,v_4)$. Despite this change in the vertex mapping, we consider these as the same embedding because the hyperedge mapping remains unchanged.

\begin{table}[t]
    \caption{Frequently Used Notations}
    \label{tab:notations}
    \centering
    \resizebox{1.0\linewidth}{!}{
    \begin{tabular}{ll}
    \toprule
    Notation & Definition \\\midrule
    $q, H$ & Query hypergraph and data hypergraph \\ 
    $V_h, E_h, L_h$ & Vertices, hyperedges, labels of a hypergraph $h$ \\
    $M$ & Embedding or partial embedding \\
    $C(e)$ & Set of candidate hyperedges for $e$ \\
    $\abs{e}$ & Arity of a hyperedge $e$ in a hypergraph $h$ \\
    $Sig(\{v_1, v_2,\dots, v_n\})$ & The multiset of labels of vertices $v_i$'s\\
    \bottomrule
    \end{tabular}
    }
\end{table}

\paragraph{Problem Statement} Given a query hypergraph $q$ and a data hypergraph $H$, the \emph{hypergraph pattern matching problem} is to find all subhypergraph embeddings of $q$ in $H$. 

\subsection{Related Works}
\paragraph{Hypergraph Pattern Matching}
Ha et al. \cite{IHSfilter} extends the subgraph matching algorithm, Turbo\textsubscript{ISO} \cite{TurboIso}, to hypergraph pattern matching, introducing a technique called the incident hyperedge structure filter.
Their algorithm generates candidates that pass the filter for each query vertex prior to matching and verifies each candidate during the matching phase. 

In contrast to previous approaches that map query vertices to data vertices in backtracking, \HGMatch \cite{HGMatch} introduces the Match-by-Hyperedge framework. This framework directly maps query hyperedges to data hyperedges, thereby reducing redundant computations for enumerating vertex mappings. Instead of generating candidates prior to matching, \HGMatch generates candidates for each hyperedge and then verifies partial embeddings extended by these candidates during the matching phase. \HGMatch proposes a necessary and sufficient condition for verification, called the equivalence of vertex profiles.
\OHMiner \cite{OHMiner} presents an overlap-centric system for hypergraph pattern mining. It constructs a DAG called an Overlap Intersection Graph (OIG) by computing intersections of hyperedges to group vertices that share incident hyperedges.
During matching, \OHMiner maintains an OIG for partial embeddings to verify the partial embedding.
%In practice, however, many hypergraph pattern matching algorithms occasionally produce incorrect results or segmentation faults due to implementation issues.

There are also problems similar to hypergraph pattern matching, with a different definition of an embedding.
\SubHyMa \cite{SubHyMa} defines an embedding as a vertex mapping, instead of a hyperedge mapping. It presents several filtering techniques that eliminate invalid candidates for query vertices by utilizing the structural properties of hyperedges and the co-occurrence of vertex pairs within them.
\HyperISO \cite{HyperIso} deals with the problem of subhypergraph containment, in which the vertices in a query hyperedge can be a subset of the vertices in the mapped data hyperedge.
It generates hyperedge candidates based on the arity and the number of adjacent hyperedges of a query hyperedge,
and it has filtering and matching phases.

\paragraph{Subgraph Matching}
Subgraph matching is a special case of hypergraph pattern matching in which an edge connects only two vertices. A subgraph matching algorithm can be used as a subroutine of hypergraph pattern matching or it can be extended to solve hypergraph pattern matching.

Numerous subgraph matching algorithms have been proposed \cite{TurboIso, VF2, VCSubgraphMatching, CFLMatch, DAF, VEQ, BICE, IVE, GuP}. 
Recent works have successfully employed the filtering-backtracking approach, utilizing an auxiliary data structure containing a set of candidate vertices for each query vertex. In the filtering phase, these algorithms remove invalid candidates from the data structure.
During backtracking, the query vertices are iteratively mapped to one of their candidates according to a matching order.
Apart from backtracking algorithms, join-based methods \cite{WCOSubgraphQueryProcessing, RapidMatch} for subgraph matching have also been studied, modeling subgraph matching as a join query in a relational database.
Interested readers can refer to the extensive surveys of recent subgraph matching algorithms \cite{Indepth, ComprehensiveSurvey}.

Subgraph matching and related problems have been extended to property graphs, where each vertex and edge can have multiple labels and key-value pairs (properties) \cite{han2024implementation, xie2024vertexsurge, van2022general}. Subgraph matching on property graphs has been of great interest in graph database systems \cite{pang2025unified} as it can express rich semantics using graph-structured data.
Unlike property graphs which have pairwise edges, hypergraphs allow edges to connect multiple vertices simultaneously, requiring fundamentally different matching strategies.

\section{Overview}
\label{sec:Overview}
\begin{algorithm}[t]
\setstretch{0.93}
    \caption{\MaCH}
    \label{Alg:Framework}
    \rm
    \SetKwProg{Fn}{Function}{:}{}
    \SetKwInOut{Input}{Input}
    \SetKwInOut{Output}{Output}
    \Input{Query hypergraph $q$, data hypergraph $H$}
    \Output{All embeddings of $q$ in $H$}
    $C_{ini}\gets$ BuildCandidateHyperedgeSpace($q, H$)\\
    Filtering($q,C_{ini}$)\\
    $M\gets \emptyset$\\
    MatchAndFilter($q,H,C, M$)\\
\end{algorithm}

\paragraph{Constraints for Hypergraph Pattern Matching}
We introduce two novel constraints, the \emph{connectivity constraint} and the \emph{intersection constraint}.
These constraints are utilized for \emph{filtering}, which is to remove candidates from the candidate hyperedge space (CHS) (\Cref{sec:Constraints}).

\paragraph{Building Candidate Hyperedge Space and Filtering}
Recent algorithms for subgraph matching use an auxiliary data structure to store candidate data vertices for each query vertex.
However, for hypergraph pattern matching, applying such data structures can be inefficient due to two primary factors below.
\begin{enumerate}[leftmargin=*]
    \item Vertex automorphisms:
    Hyperedges often contain multiple vertices that have the same label, leading to redundant computations when matching vertices individually.
    \item Relationships among multiple vertices: A hyperedge connects multiple vertices simultaneously, extending the pairwise connections in traditional graphs. Vertex-based data structures are not suitable for effectively utilizing such relationships in hypergraphs.
\end{enumerate}
To address these limitations, we introduce a \emph{candidate hyperedge space} (CHS), a data structure storing candidates for query hyperedges rather than query vertices. CHS stores potential mappings between hyperedges in the query hypergraph and the data hypergraph. We then apply the connectivity constraint to filter out invalid candidates (which cannot be included in any embeddings) in the CHS, significantly reducing the search space (\Cref{sec:CandidateHyperedgeSpace}).

\paragraph{Match-and-Filter Framework}
Although initial filtering is quite effective, it cannot fully capture the complex connectivity patterns of hypergraphs. These patterns, where multiple vertices are connected simultaneously by three or more hyperedges, often become apparent only during the matching process.
To address this, we introduce a novel match-and-filter framework, which integrates filtering operations directly into the matching process.
This framework interleaves matching and filtering operations, utilizing the current partial embedding to filter out invalid candidates from CHS during the matching process.

\begin{enumerate}[leftmargin=*]
    \item We iteratively extend the current partial embedding by mapping a query hyperedge to a candidate data hyperedge in the CHS.
    \item As the new mapping is added, we apply our connectivity and intersection constraints to identify and eliminate candidates which are incompatible with the current partial embedding.
\end{enumerate}
By integrating filtering operations directly into the matching process, our framework efficiently handles the complex connectivity patterns in hypergraphs (\Cref{sec:MatchAndClean}).

\Cref{Alg:Framework} shows the outline of our algorithm \MaCH, which takes a query hypergraph $q$ and a data hypergraph $H$ as input, and finds all embeddings of $q$ in $H$.

\section{Constraints for Hypergraph Pattern Matching}
\label{sec:Constraints}
In this section, we introduce these constraints: Hyperedge signature constraint, Connectivity constraint, and Intersection constraint.

\subsection{Hyperedge Signature Constraint}

We begin by defining the concept of a signature for a set of vertices, which is fundamental to our constraints.

\begin{definition}
    For a set of vertices $S\subseteq V$, the signature of $S$, denoted as $Sig(S)$, is defined as the multiset of labels of $v\in S$, i.e., $Sig(S) = multiset\{L(v)\mid v\in S\}$.
\end{definition}

This definition extends the concept of signature introduced by \HGMatch \cite{HGMatch}, which originally defined it only for hyperedges. Our generalization allows us to apply this concept to arbitrary sets of vertices, which will be crucial for our connectivity constraint and intersection constraint.
For simplicity of notation, we will use set notation, e.g. $\{A,A,B\}$, to represent a signature, although it is a multiset.
Using this definition, we can now state the Hyperedge Signature Constraint.

\begin{lemma}[Hyperedge Signature Constraint]
Given a query hypergraph $q$ and a data hypergraph $H$, a hyperedge $e \in E_q$ can be mapped to a hyperedge $f \in E_H$ only if $e$ and $f$ have identical signatures.
\end{lemma}

\subsection{Connectivity Constraint}

Matching individual hyperedges based solely on their signatures is not sufficient to ensure a valid partial embedding.
To check pairwise relationships between hyperedges, we introduce the concept of connectivity constraint.

\begin{lemma}[Connectivity Constraint]
Given a query hypergraph $q$ and a data hypergraph $H$, for any pair of adjacent query hyperedges $e, e' \in E_q$ and data hyperedges $f, f' \in E_H$, a mapping $\{(e,f), (e',f')\}$ can be a partial embedding only if $Sig(e \cap e') = Sig(f \cap f')$.
\end{lemma}

\begin{example}
    Consider the hypergraphs in \Cref{fig:hypergraphs} and a hyperedge mapping $M=\{(e_1,f_2), (e_2,f_3)\}$. This mapping satisfies the hyperedge signature constraint as $Sig(e_1)=Sig(f_2)=\{A,A,A,B,B\}$ and $Sig(e_2)=Sig(f_3)=\{A,A,A\}$. However, it fails to meet the connectivity constraint because $Sig(e_1\cap e_2)=\{A,A\}$ while $Sig(f_2\cap f_3)=\{A\}$. Therefore, $M$ is not a valid partial embedding.
\end{example}

\subsection{Intersection Constraint}
\label{subsec:IntersectionConstraint}

While the hyperedge signature constraint ensures valid individual hyperedge mappings and the connectivity constraint checks pairwise relationships between hyperedges, these are not sufficient to guarantee a valid embedding as they fail to capture the complex connectivity patterns involving three or more hyperedges.

\begin{restatable}[Intersection Constraint]{theorem}{ThmIntersectionConstraint}
\label{thm:IntersectionSignature}
    Given a query hypergraph $q$, a data hypergraph $H$, and a hyperedge mapping $M: E_{q} \to E_H$, $M$ is an embedding if and only if $Sig(\bigcap_{e\in S} e) = Sig(\bigcap_{e\in S} M(e))$ for every subset $S\subseteq E_{q}$.
\end{restatable}

\begin{example}
    Consider the hypergraphs in \Cref{fig:hypergraphs} and a hyperedge mapping $M = \{(e_1, f_1), (e_2,f_3), (e_3, f_8), (e_4,f_4)\}$ shown in \Cref{subfig:wrongmapping}. This mapping satisfies both the hyperedge signature constraint for each hyperedge and the connectivity constraint for every pair of adjacent hyperedges. For instance, $Sig(e_1 \cap e_2) = Sig(f_1 \cap f_3) = \{A,A\}$. However, $Sig(e_1 \cap e_2 \cap e_4)=\{A\}$ as it includes only one vertex $u_2$, while $Sig(f_1 \cap f_3 \cap f_4)=\{A,A\}$ as it includes two vertices $v_1$ and $v_2$. Thus, $M$ is not a valid embedding.
\end{example}

Unlike previous constraints, the intersection constraint provides a necessary and sufficient condition for a mapping to be an embedding. This stronger constraint allows us to precisely characterize valid embeddings.
Building on this, we introduce the concept of $M$-compatibility for candidates.

\begin{definition}
Let $M$ be a partial embedding for a partial query $q'$ of $q$ and let $e \in E_q \setminus E_{q'}$ be an unmapped query hyperedge. A candidate $f$ for $e$ is \emph{$M$-compatible} if $M \cup \{(e, f)\}$ is a partial embedding for the partial query induced by $E_{q'} \cup \{e\}$.
\end{definition}

\section{Candidate Hyperedge Space}
\label{sec:CandidateHyperedgeSpace}

\begin{figure}
    \centering
    \begin{subfigure}{0.44\linewidth}
        \centering
        \includegraphics[width=\linewidth,trim={7mm 0 7mm 0}]{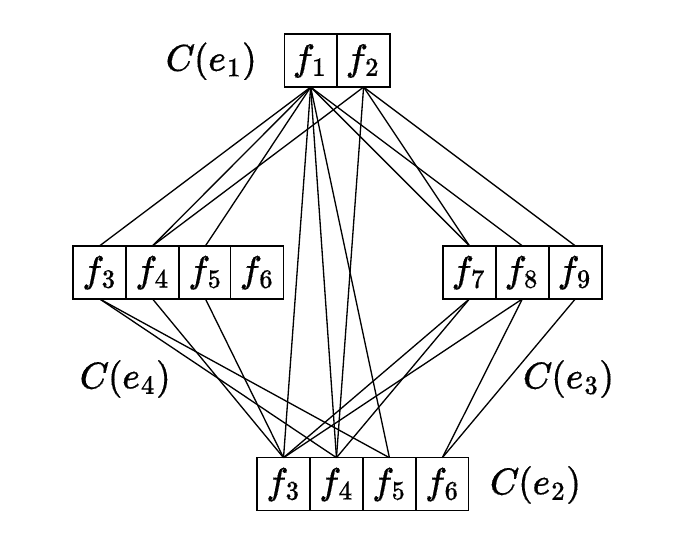}
        \caption{Initial CHS}
        \label{subfig:InitialCHS}
    \end{subfigure}
    \begin{subfigure}{0.44\linewidth}
        \centering
        \includegraphics[width=\linewidth,trim={7mm 0 7mm 0}]{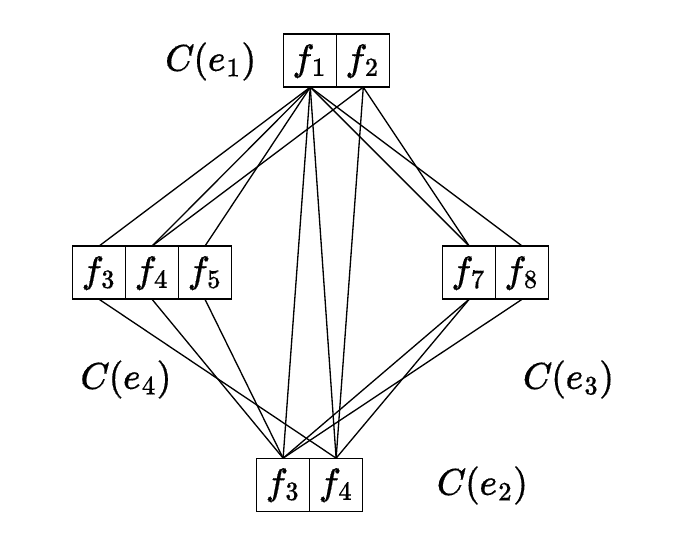}
        \caption{Filtered CHS}
        \label{subfig:RefinedCHS}
    \end{subfigure}
    \caption{Candidate hyperedge space on $q$ and $H$ in \Cref{fig:hypergraphs} before and after filtering}
    \label{fig:FilteringProcess}
\end{figure}

\subsection{Building Candidate Hyperedge Space}
\begin{definition}
Given a query hypergraph $q$ and a data hypergraph $H$, a \emph{Candidate Hyperedge Space (CHS)} on $q$ and $H$ consists of the candidate hyperedge set $C(e)$ for each $e\in E_q$ and connections between the candidates as follows.
\begin{enumerate}[leftmargin=*]
    \item For each $e\in E_q$, candidate hyperedge set $C(e)$ is a set of hyperedges in $H$ that $e$ can be mapped to.
    \item There is a connection between $f\in C(e)$ and $f'\in C(e')$ if $e$ and $e'$ are adjacent in $q$, $f$ and $f'$ are adjacent in $H$, and $Sig(e\cap e') = Sig(f\cap f')$.
\end{enumerate}
\end{definition}

We initialize the CHS as follows. For each query hyperedge $e$, initial candidate hyperedge set $C_{ini}(e)$ is defined as the set of data hyperedges $f$ which satisfy $Sig(f)=Sig(e)$.
We then establish connections between candidate hyperedges. For each pair of adjacent query hyperedges $e$ and $e'$, we create a connection between $f \in C(e)$ and $f' \in C(e')$ if $f$ and $f'$ are adjacent and $Sig(e \cap e') = Sig(f \cap f')$.

\begin{example}
    \Cref{subfig:InitialCHS} shows the initial CHS $C_{ini}$ for query hypergraph $q$ and data hypergraph $H$ in \Cref{fig:hypergraphs}. For $e_1$, $Sig(e_1) = \{A,A,A,B,B\}$. As $f_1$ and $f_2$ have the same signature as $e_1$, $C_{ini}(e_1) = \{f_1, f_2\}$. Similarly, candidate hyperedge sets for other query hyperedges include data hyperedges with matching signatures.

    In $q$, $e_1$ and $e_2$ are adjacent with $Sig(e_1 \cap e_2) = \{A,A\}$. Candidates for $e_1$ and $e_2$ have a connection if their intersections match this signature. For instance:
    \begin{itemize}[leftmargin=*]
        \item $f_1\in C(e_1)$ and $f_3\in C(e_2)$ are connected as $Sig(f_1 \cap f_3) = \{A,A\}$,
        \item $f_1\in C(e_1)$ and $f_6\in C(e_2)$ are not connected as $f_1 \cap f_6 = \emptyset$,
        \item $f_2\in C(e_1)$ and $f_3\in C(e_2)$ are not connected as $Sig(f_2 \cap f_3) = \{A\}$.
    \end{itemize}
    Note that $e_4$ and $e_3$ are not adjacent in $q$, so their candidates do not have connections in the CHS, regardless of their intersections in $H$.
\end{example}

All hyperedges constituting an embedding are present in the CHS. We call this property of CHS as \textit{complete}.

The number of candidates in CHS is $O(\abs{E_q}\abs{E_H})$, and the number of connections is $O(\abs{E_q}^2\abs{E_H}^2)$ when all hyperedges in the query hypergraph and all hyperedges in the data hypergraph are adjacent to each other, respectively.
But, in practice these numbers are much smaller.

We define adjacent candidate set $C(e'\mid e,f)$ as the set of hyperedges in $H$ that $e'$ can be mapped to when $e$ is mapped to $f$. That is, $C(e'\mid e,f)$ is the set of $f'\in C(e')$ which is connected to $f\in C(e)$.
    
\subsection{Filtering by Connectivity Constraint}

\begin{algorithm}[t]
\setstretch{0.93}
    \caption{\textsc{Filtering($q, C$)}}
    \label{Alg:DefaultFiltering}
    \rm
    \SetKwProg{Fn}{Function}{:}{}
    \SetKwInOut{Input}{Input}
    \SetKwInOut{Output}{Output}
    \Input{Query hypergraph $q$ and CHS $C$}
    Initialize queue $Q$ with all pairs of adjacent query hyperedges\\
    \While{\rm $Q$ is not empty} {
        $(e, e')$ $\gets$ Q.pop()\\
        $removed \gets false$\\
        \ForEach{$f\in C(e)$}{
            \If{\rm $C(e'\mid e,f) = \emptyset$} {
                Remove $f$ from $C(e)$\\
                $removed \gets true$\\
            }
        }
        \If{removed}{
            \ForEach{$e''$ adjacent to $e$}{
                \If{$e'' \neq e' \And (e'',e) \notin Q$}{
                    $Q.push((e'', e))$
                }
            }
        }
    }
\end{algorithm}

Before the matching phase, we apply the connectivity constraint to the initial CHS as a filtering method to remove candidate hyperedges that cannot be included in an embedding.
For a query hyperedge $e$ to be mapped to a candidate $f$, $f$ must have at least one connected candidate in $C(e')$ for each query hyperedge $e'$ adjacent to $e$. Formally, $C(e' \mid e,f) \neq \emptyset$ should hold for every $e'$ adjacent to $e$.

In \Cref{Alg:DefaultFiltering}, we iteratively remove candidates that fail to meet the connectivity constraint until the CHS contains no candidates that violate the constraint.
The algorithm begins by initializing a queue with all pairs of adjacent query hyperedges. Then it iterates through these pairs $(e, e')$ in the queue, examining the candidates for hyperedge $e$. For each candidate $f \in C(e)$, if $f$ has no connected candidates in $C(e')$ (i.e., $C(e'\mid e,f)$ is empty), $f$ is removed from $C(e)$.
When candidates are removed from $C(e)$, the algorithm propagates these changes by adding new query hyperedge pairs $(e'', e)$ to the queue for $e''$ adjacent to $e$. The algorithm terminates when the queue is empty, indicating that no further filtering is possible.

\begin{example}
    \Cref{fig:FilteringProcess} shows the initial CHS and the CHS after filtering by connectivity constraint.
    In the initial CHS, $f_6\in C(e_4)$ has no connection to candidates of $e_1$ (and also $e_2$), so we remove $f_6$ from $C(e_4)$. Similarly, $f_5\in (e_2)$ has no connection to candidates of $e_3$, so we remove $f_5$ from $C(e_2)$. $f_6\in C(e_2)$ has no connection to candidates of $e_1$, so we remove $f_6$ from $C(e_2)$. After removing $f_6$ from $C(e_2)$, $f_9\in C(e_3)$ loses its only connection to $e_2$'s candidates, so we remove $f_9$ from $C(e_3)$.
    In the filtered CHS, every remaining candidate has at least one connection to candidates of every adjacent query hyperedge.
\end{example}

\begin{restatable}[]{theorem}{ThmFilteringtIme}
\label{thm:FilteringtIme}
     \Cref{Alg:DefaultFiltering} takes $O((\sum_e\abs{C(e)})^2)$ time.
\end{restatable}

While the number of candidates $\abs{C(e)}$ can reach $\abs{E_H}$ in the worst case, it is typically much smaller than $\abs{E_H}$ in practice.

\section{Match-and-Filter}
\label{sec:MatchAndClean}

\begin{algorithm}[t]
\setstretch{0.93}
    \caption{MatchAndFilter($q,H,C, M$)}
    \label{Alg:Backtracking}
    \rm
    \SetKwProg{Fn}{Function}{:}{}
    \SetKwInOut{Input}{Input}
    \Input{Query hypergraph $q$, data hypergraph $H$, CHS $C$, partial embedding $M$}
    \If{$\exists e \in E_q$ s.t. $C(e) = \emptyset$}{
        \Return \tcp{No valid embedding}
    }
    $e'\gets$ MatchingOrder($q, C, M$)\\
    \ForEach{$f'\in C(e')$}{
        $M'\gets M\cup \{(e',f')\}$\\
        Update $b_q$ and $b_H$\\
        \label{line:updateIHB}
        \If{$\abs{M'}=\abs{E_q}$}{
            Report $M'$\\
        }
        \Else{
            \label{line:StartConn}
            \tcp{Connectivity Constraint}
            \ForEach{unmapped $e\in E_q$ adjacent to $e'$} {
                \ForEach{$f\in C(e)$}{
                    \If{\rm $f\in C(e)$ is not connected to $f'\in C(e')$} {
                        Remove $f$ from $C(e)$\\
                    }
                }
            \label{line:EndConn}
            }
            \tcp{Intersection Constraint}
            \ForEach{unmapped $e\in E_q$} {
                \ForEach{$f\in C(e)$}{
                    CheckIntersection($q,H,C,M',e,f$)\\
                }
            }
            MatchAndFilter($q, H, C, M'$)
        }
        Restore $b_q$ and $b_H$\\
    }
    \Return
\end{algorithm}

After finishing initial filtering described in the previous section, \Cref{Alg:Backtracking} finds subhypergraph embeddings using a dynamic matching order. It extends the partial embedding $M$ to $M\cup \{(e',f')\}$ by mapping a query hyperedge $e'$ to its candidate hyperedge $f'$. After each extension, it applies a two-stage filtering process to maintain $M$-compatibility for all candidates in the candidate hyperedge space.

\begin{enumerate}[leftmargin=*]
    \item Connectivity Constraint Check: A fast, coarse-grained filtering technique that quickly eliminates candidates incompatible with the current partial embedding in $O(1)$ time per candidate (lines 10--13). 
    \item Intersection Constraint Check: A more thorough, fine-grained filtering, ensuring that all remaining candidates are $M$-compatible. As we will show in \Cref{thm:IStIme}, this check takes $O(\abs{e})$ time for each pair of an unmapped query hyperedge $e$ and its candidate $f\in C(e)$ (lines 14--16).
\end{enumerate}

This two-stage process enables early removal of invalid candidates by the quick connectivity constraint check before applying the more expensive intersection constraint check. It achieves improved overall efficiency, when compared to applying only the intersection constraint check, which also obtains the same $M$-compatibility.

If the candidate set of any query hyperedge becomes empty after the two filtering stages, the current partial embedding cannot be extended to a valid embedding, and the algorithm backtracks to try the next candidate. Otherwise, the algorithm continues with the extended partial embedding. 

\subsection{Filtering by Connectivity Constraint}

\begin{figure}
    \centering

    \begin{subfigure}{0.44\linewidth}
        \centering
        \includegraphics[width=\linewidth,trim={7mm 0 7mm 0}]{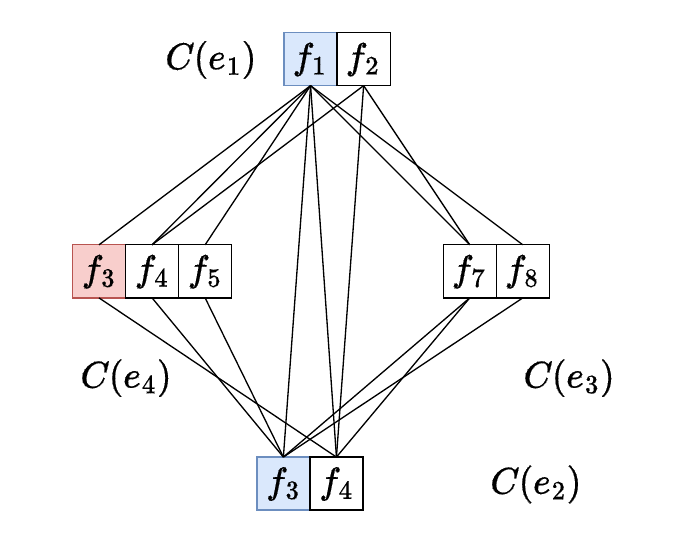}
        \caption{Connectivity constraint}
        \label{subfig:CleanByConnectivity}
    \end{subfigure}
    \begin{subfigure}{0.44\linewidth}
        \centering
        \includegraphics[width=\linewidth,trim={7mm 0 7mm 0}]{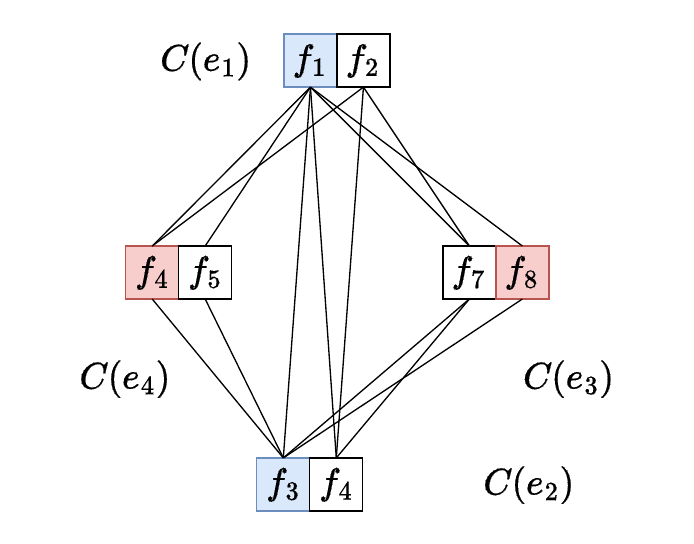}
        \caption{Intersection constraint}
        \label{subfig:CleanByIntersection}
    \end{subfigure}
    \caption{Match-and-Filter applied to CHS on $q$ and $H$ in \Cref{fig:hypergraphs}. Current partial embedding $M$ is $\{(e_1, f_1), (e_2, f_3)\}$ (blue). Incompatible candidates (red) are filtered by connectivity constraint in (a) and intersection constraints in (b).}
    \label{fig:MappingProcess}
\end{figure}

Lines 10--13 of \Cref{Alg:Backtracking} are to apply the connectivity constraint to the candidates in CHS after a new hyperedge pair $(e',f')$ is added to the partial embedding $M$. 

The algorithm iterates through all unmapped query hyperedges $e$ that are adjacent to the newly mapped hyperedge $e'$. For each such $e$, it examines $C(e)$ and removes any candidate $f$ that is not in $C(e\mid e',f')$.
For each candidate $f\in C(e)$, it takes $O(1)$ time to check whether $f\in C(e)$ is connected to the newly mapped hyperedge pair $(e', f')$, as we already built connections in the CHS.

\begin{example}
    Consider \Cref{subfig:CleanByConnectivity} with a partial embedding $M=\{(e_1, f_1), (e_2, f_3)\}$, where $(e_2, f_3)$ is a newly mapped hyperedge pair. We remove invalid candidates from the candidate sets of unmapped query hyperedges by the connectivity constraint.
    $f_3$ is removed from $C(e_4)$ because there is no connection between the newly mapped pair $(e_2, f_3)$ and the candidate $f_3\in C(e_4)$. In contrast, $f_4$ and $f_5$ remain in $C(e_4)$ as they have connections to $f_3\in C(e_2)$.
\end{example}

\subsection{Filtering by Intersection Constraint}
\label{subsec:CleaningByIntersection}

\begin{algorithm}[t]
\setstretch{0.93}
    \caption{\textsf{CheckIntersection($q, H, C, M, e, f$)}}
    \label{Alg:CheckIntersectionSignature}
    \rm
    \SetKwInOut{Input}{Input}
    \SetKwProg{Fn}{Function}{:}{}   
    
    \Input{Query $q$, data $H$, CHS $C$, partial embedding $M$, query hyperedge $e$, data hyperedge $f$}
    \ForEach{$u\in e$}{
        $b_q[u] \gets b_q[u] + 2^{pos(e)}$\\
    }
    \ForEach{$v\in f$}{
        $b_H[v] \gets b_H[v] + 2^{pos(e)}$\\
    }
    $S_q^+\gets [(b_{q}[u], L_q[u]) \mid u \in e]$\\
    $S_H^+\gets [(b_{H}[v], L_H[v]) \mid v \in f]$\\
    
    \If{$sort(S_q^+) \neq sort(S_H^+)$}{
        Remove $f$ from $C(e)$\\
    }
    Restore IHBs\\
\end{algorithm}

\Cref{Alg:CheckIntersectionSignature} is to apply the intersection constraint to the candidate $f\in C(e)$ for the partial embedding $M$.
It filter out the candidate from CHS if it is not $M$-compatible.
To check the $M$-compatibility of a candidate, we use \Cref{thm:IntersectionSignature}, which states that an embedding is valid if and only if, for every subset $S$ of query hyperedges, the signature of the intersection of $S$ matches the signature of the intersection of the corresponding data hyperedges for $S$. However, computing the signature for the intersection of every subset can take exponential time.

\begin{figure}[t]
    \centering
    \begin{subfigure}{0.32\linewidth}
    \centering
        \begin{tikzpicture}
            \begin{scope}
                \clip (0,0) circle (0.8);
                \fill[blue!10] (1,0) circle (0.8);
            \end{scope}
            \begin{scope}
                \clip (0,0) circle (0.8);
                \clip (0.5, 0.86) circle (0.8);
                \fill[blue!40] (1,0) circle (0.8);
            \end{scope}
            \draw (0,0) circle (0.8) node [below left] {$e_2$};
            \draw (1,0) circle (0.8) node [below right] {$e_3$};
            \draw (0.5,0.86) circle (0.8) node [above] {$e_1$};
        \end{tikzpicture}
        \caption{Intersection $e_2 \cap e_3$}
        \label{subfig:SignatureExample}
    \end{subfigure}
    \begin{subfigure}{0.32\linewidth}
    \centering
        \begin{tikzpicture}
            \begin{scope}
                \clip (0,0) circle (0.8);
                \fill[blue!10] (1,0) circle (0.8);
            \end{scope}
            \begin{scope}
                \clip (0,0) circle (0.8);
                \clip (0.5,0.86) circle (0.8);
                \fill[white] (1,0) circle (0.8);
            \end{scope}
            \draw (0,0) circle (0.8) node [below left] {$e_2$};
            \draw (1,0) circle (0.8) node [below right] {$e_3$};
            \draw (0.5,0.86) circle (0.8) node [above] {$e_1$};
        \end{tikzpicture}
        \caption{Cell of $\{e_2, e_3\}$}
        \label{subfig:CellExample}
    \end{subfigure}
    \begin{subfigure}{0.32\linewidth}
    \centering
        \begin{tikzpicture}
            \begin{scope}
                \clip (0,0) circle (0.8);
                \fill[blue!10] (1,0) circle (0.8);
            \end{scope}
            \begin{scope}
                \clip (0,0) circle (0.8);
                \clip (0.5,0.86) circle (0.8);
                \fill[white] (1,0) circle (0.8);
            \end{scope}
            \draw (0,0) circle (0.8) node [below left] {$f_1$};
            \draw (1,0) circle (0.8) node [below right] {$f_3$};
            \draw (0.5,0.86) circle (0.8) node [above] {$f_2$};
        \end{tikzpicture}
        \caption{Cell of $\{f_1, f_3\}$}
        \label{subfig:CellExample_data}
    \end{subfigure}
    \caption{Intersection and cells of hyperedges for the partial embedding $\{(e_1, f_2), (e_2, f_1), (e_3,f_3)\}$}
    \label{fig:IntersectionAndCell}
\end{figure}

To address this, we shift our focus from intersections to \emph{cells}. 
Consider a set $E$ of query hyperedges. For a subset $S$ of $E$, the cell of $S$ is defined as the set of vertices which are only incident to hyperedges in $S$ and not incident to hyperedges in $E\setminus S$.
For example, in \Cref{subfig:SignatureExample}, the intersection of $e_2$ and $e_3$ can be divided into two cells: the cell of $\{e_1,e_2,e_3\}$, representing $e_1 \cap e_2 \cap e_3$ (in dark blue), and the cell of $\{e_2, e_3\}$, representing $e_2 \cap e_3 \setminus e_1$ (in light blue).

A cell can be represented by a bitmap, where each bit position corresponds to a hyperedge in $E$. In \Cref{fig:IntersectionAndCell}, let's assume that  $pos(e_1)$ (i.e., bit position of $e_1$) is 0 (rightmost position), $pos(e_2)=1$, and $pos(e_3)=2$.
In the bitmap of a cell $S$, the bit position of a hyperedge $e$ has 1 if and only if $e$ is in $S$, i.e., the colored cell in \Cref{subfig:CellExample} is represented as the bitmap 110. (Note that the bitmap is the bit representation of integer $\sum_{e \in S} 2^{pos(e)}$.)

For a query hyperedge set $S$ and its corresponding data hyperedge set $M(S)$ under a partial embedding $M$, the cell of $M(S)$ is represented by the same bitmap as the cell of $S$. 
For example, in \Cref{subfig:CellExample_data}, the cell of $\{f_1, f_3\}$ is represented by the bitmap 110,
where the mapping is $\{(e_1, f_2), (e_2, f_1), (e_3,f_3)\}$.

\begin{restatable}{theorem}{ThmHILCequivalence}
\label{thm:HILCequivalence}
    Given a partial query $q'$ of $q$ and hyperedge mapping $M: E_{q'} \to E_H$, $M$ is a partial embedding
    if and only if the signature of a cell $S$ matches the signature of $M(S)$ for every subset of query hyperedges $S\in E_{q'}$.
\end{restatable}

OHMiner \cite{OHMiner} gives an observation which is similar to the above theorem. However, our approach for verification is completely different from OHMiner. While OHMiner computes set intersections repeatedly for vertex sets to verify a candidates, which incurs significant computational overhead, our approach does not utilize set intersections.

Instead, we introduce two new auxiliary data structures, the \emph{Incident Hyperedge Bitmap} (IHB) and the \emph{Cell Signature} to compute the signatures of cells. These structures enable us to perform $M$-compatibility checks for each candidate $f\in C(e)$ in $O(\abs{e})$ time, where $\abs{e}$ is the arity of the query hyperedge $e$.

\begin{definition}[Incident Hyperedge Bitmap]
For a vertex $u \in V_q$, the Incident Hyperedge Bitmap $b_q^M(u)$ is the bitmap of the cell that $u$ belongs to.
Similarly, for a vertex $v \in V_H$, IHB $b_H^M(v)$ is the bitmap of the cell that $v$ belongs to.
\end{definition}
Both $b_q^M$ and $b_H^M$ can be computed incrementally when a new mapping $(e,f)$ is added to a partial embedding $M$.

\begin{example}
    Consider the hypergraphs in \Cref{fig:hypergraphs}, a hyperedge mapping $M = \{(e_1, f_1), (e_2,f_3), (e_3, f_8), (e_4,f_4)\}$ shown in \Cref{subfig:wrongmapping}, and bit positions of $e_1, e_2, e_3, e_4$ are $0,1,2,3$, respectively.
    \begin{itemize}[leftmargin=*]
        \item $b_q^M(u_1) = 0111$ because $u_1$ is in the cell of $\{e_1, e_2, e_3\}$.
        \item $b_q^M(u_2) = 1011$ because $u_2$ is in the cell of $\{e_1, e_2, e_4\}$.
        \item $b_q^M(u_3) = 1001$ because $u_3$ is in the cell of $\{e_1, e_4\}$.
        \item $b_H^M(v_1) = 1011$ because $v_1$ is in the cell of $\{f_1, f_3, f_4\}$, whose elements are mapped to $e_1, e_2, e_4$ respectively.
        \item $b_H^M(v_2) = 1011$ because $v_2$ is in the cell of $\{f_1, f_3, f_4\}$, whose elements are mapped to $e_1, e_2, e_4$ respectively.
        \item $b_H^M(v_3) = 0101$ because $v_3$ is in the cell of $\{f_1, f_8\}$, whose elements are mapped to $e_1, e_3$ respectively.
        \end{itemize}
\end{example}

In addition, we introduce Cell Signatures, $\I_q^M(b, l)$ and $\I_H^M(b, l)$.
Given a partial query $q'$ of $q$ and hyperedge mapping $M: E_{q'}\to E_H$, 
$\I_q^M(b, l)$ (resp. $\I_H^M(b, l)$) stores the signature of  the cell in the domain of $M$ (resp. the image of $M$) represented by bitmap $b$ by counting the number of vertices with label $l$ in the cell.

\begin{example}
        Consider the hypergraphs in \Cref{fig:hypergraphs}, a hyperedge mapping $M = \{(e_1, f_1), (e_2,f_3), (e_3, f_8), (e_4,f_4)\}$ shown in \Cref{subfig:wrongmapping}.
        \begin{itemize}[leftmargin=*]
            \item $\I_q^M(1011,A)=1$ since $u_2$ is the only vertex whose IHB is $1011$ and label is $A$.
            \item $\I_q^M(1001,A)=1$ since $u_3$ is the only vertex whose IHB is $1001$ and label is $A$. Although $u_2$ is also in $e_1\cap e_4$, $u_2$ is incident to $e_2$ as well, which makes its IHB $1011$.
            \item $\I_H^M(1011,A)=2$ since there are two vertices $v_1$ and $v_2$ whose IHB is $1011$ and label is $A$.
        \end{itemize}
\end{example}

For a partial embedding $M$, to check $M$-compatibility of a candidate $f\in C(e)$, it is sufficient to compare the signatures of cells by \Cref{thm:HILCequivalence}. Furthermore, since $M$ is already a valid partial embedding, we only have to check signatures of cells that have changed due to the addition of mapping $(e,f)$.
Thus, instead of checking all entries in $\I_q^{M'}=\I_H^{M'}$ for $M'=M\cup\{(e,f)\}$, we only need to consider the signatures of the cells that are subsets of $e$.

\begin{restatable}{theorem}{ThmISComplexity}
\label{thm:IScomplexity}
Let $M: E_{q'} \to E_H$ be a partial embedding for a partial query $q'$ of $q$.   A candidate $f \in C(e)$ for a query hyperedge $e \in E_q \setminus E_{q'}$
is $M$-compatible 
if $\I_q^{M'}(b,l)=\I_H^{M'}(b,l)$ for the bitmap $b$ of every cell $S$ that is a subset of $e$ and every label $l$.
\end{restatable}

\begin{example}
    Consider \Cref{subfig:CleanByIntersection} with a partial embedding $M=\{(e_1, f_1),(e_2, f_3)\}$. We remove invalid candidates the candidate sets of unmapped query hyperedges using the intersection constraint. 
    We check for a potential extension $M'=\{(e_1, f_1),(e_2, f_3),(e_4,f_4)\}$.
    For this extension, we find that $\I_q^{M'}(1011, A) = 1$ while $\I_H^{M'}(1011, A) = 2$.
    This difference in the cell signatures shows that $f_4$ cannot be mapped to $e_4$, so we remove it from $C(e_4)$. Similarly, $f_8$ is removed from $C(e_3)$. After filtering, each of $e_3$ and $e_4$ has only one remaining candidate. Mapping these candidates results in the embedding shown in \Cref{subfig:embedding}.
\end{example}

\Cref{Alg:CheckIntersectionSignature} applies the intersection constraint to the candidate $f\in C(e)$, and it removes $f$ if $f$ is not $M$-compatible.
Although there are $2^{\abs{E_q}}$ bitmaps in the definition of the cell signature $\I_q^{M'}$ for mapping $M'=M\cup \{(e,f)\}$, at most $\abs{e}$ entries of the cell signature are non-zero because the hyperedge $e$ has $\abs{e}$ vertices, and each vertex in $e$ can be included in only one cell.
Thus, instead of computing the cell signature for every bitmap of a cell that is a subset of $e$, we create an array of non-zero entries, i.e., $(b^{M'}_q[u], L_q[u])$ for each vertex $u$ in the query hyperedge. Similarly, we create an array of $(b^{M'}_H[v], L_H[v])$ for each vertex $v$ in the data hyperedge, and check whether these two arrays are equal.

\begin{theorem}\label{thm:IStIme}
    The running time of \Cref{Alg:CheckIntersectionSignature} is $O(\abs{e})$.
\end{theorem}

\begin{proof}
    \Cref{Alg:CheckIntersectionSignature} computes $b^{M'}_q$ and $b^{M'}_H$ by considering each $u\in e$ and each $v\in f$ in lines 1--4.
    Then, it computes non-empty cells that are subsets of $e$ and $f$ ($S_q^+$ and $S_H^+$ in lines 5--6). In lines 7--8, it compares two arrays after sorting these arrays by radix sort, and removes the candidate $f$ if the arrays are not equal.
    It takes $O(\abs{e})$ time to compute IHB, to compute non-empty entries, and to sort and compare two arrays.
\end{proof}

\HGMatch also presents a necessary and sufficient condition for hypergraph pattern matching called equivalence of vertex profiles. This condition can be used to check whether a candidate $f\in C(e)$ is $M$-compatible in $O(\overline{a_q} \times \abs{E_q})$ time, where $\overline{a_q}$ is the average arity of the query hypergraph. The time complexity of our approach is an improvement by a factor of $\abs{E_q}$ compared to HGMatch.

\subsection{Matching Order}

In subgraph matching, two well-known heuristics are the degree heuristic \cite{RapidMatch, RISubgraphMatching} and the candidate size heuristic \cite{DAF, VEQ}.
The degree heuristic prioritizes query vertices with the highest degree, 
and the candidate size heuristic prioritizes query vertices with the smallest number of candidates. 
For subgraph matching on property graphs, van Leeuwen et al.~\cite{van2022general} suggest various matching orders, such as utilizing the number of constraints.

We propose a matching order that combines both the degree heuristic and the candidate size heuristic, adapted for the hypergraph context.
Line 3 of \Cref{Alg:Backtracking} selects a hyperedge based on the following criteria: 1) Unmapped hyperedges with only one candidate are selected first; 2) If no such hyperedge exists, we choose the hyperedge with the largest number of unmapped adjacent hyperedges;
3) In case of a tie, we select the hyperedge with the smallest number of candidates.

We prioritize unmapped hyperedges with only one candidate, because these hyperedges must be mapped eventually, and their early mapping immediately reduce the candidates on other hyperedges. We also give priority to hyperedges with many unmapped adjacent hyperedges, because they provide more constraints for subsequent matches. By considering both the number of unmapped adjacent hyperedges and the size of candidate sets, we can effectively prune the search space.

\section{Performance Evaluation}
\label{sec:PerformanceEvaluation}

In this section, we conduct experiments to evaluate the effectiveness of our algorithm, \MaCH. Our evaluation focuses on two main aspects. First, we compare our algorithm's performance against the state-of-the-art algorithms to demonstrate overall effectiveness of our approach. Second, we analyze individual techniques in our algorithm to show their specific contributions to performance improvement. 

\subsection{Experimental Setup}
\label{subsec:ExperimentalSetup}

\paragraph{Baseline}
We primarily compare \MaCH with \HGMatch \cite{HGMatch} because \HGMatch significantly outperforms other methods such as an extension of a subgraph matching algorithm \cite{IHSfilter}. We also include the very recently proposed \OHMiner \cite{OHMiner} in our main experiments. In addition, we compare \MaCH with the approach of transforming a hypergraph into a bipartite graph and solving the subgraph matching problem. For this approach, we use the state-of-the-art subgraph matching algorithm \GuP \cite{GuP}.

\begin{table}[t]
    \caption{Statistics of Data Hypergraphs. $\abs{V}$ and $\abs{E}$ denote the number of vertices and hyperedges, respectively. $\abs{\Sigma}$ is the size of the vertex label set (`-' indicates no labels), $a^{\max}$ is the maximum arity, and $\overline{a}$ is the average arity.}
    \centering
    \resizebox{\linewidth}{!}{
    \begin{tabular}{l r r r r r r r r r}
    \toprule
     & HC & NC & MA & NS & CH & CP & SB & EU & HB \\ 
    \midrule
    $\abs{V}$ & 1.3K & 1.2K & 73.9K & 5.6K & 327 & 242 & 294 & 1.0K & 1.5K \\
    $\abs{E}$ & 336 & 1.0K & 5.4K & 6.6K & 7.8K & 12.7K & 21.7K & 24.5K & 54.9K \\
    $\abs{\Sigma}$ & 2 & - & 704 & - & 9 & 11 & 2 & - & 2 \\
    $a^{\max}$ & 81 & 39 & 1.8K & 187 & 5 & 5 & 99 & 40 & 399 \\
    $\overline{a}$ & 35.15 & 6.17 & 24.19 & 9.70 & 2.33 & 2.42 & 9.90 & 3.62 & 22.15 \\
    \midrule
    \midrule
     & WT & CB & TC & CM & SA & CD & AR & AM & \\
    \midrule
    $\abs{V}$ & 88.9K & 1.7K & 172.7K & 1.0M & 15.2M & 1.9M & 2.3M & 13.3M & \\
    $\abs{E}$ & 66.0K & 83.1K & 221.0K & 252.7K & 1.1M & 2.2M & 4.2M & 4.8M & \\
    $\abs{\Sigma}$ & 11 & - & 160 & - & 26.3K & - & 29 & 258 & \\
    $a^{\max}$ & 25 & 25 & 85 & 925 & 61.3K & 280 & 9.4K & 2.4K & \\
    $\overline{a}$ & 6.86 & 8.81 & 3.18 & 3.02 & 23.67 & 3.43 & 17.17 & 7.40 & \\
    \bottomrule
    \end{tabular}
    }
    \label{tab:datasets}
\end{table}

\begin{figure*}[t]
    \centering
    \includegraphics[width = 0.87\textwidth]{figures/elapsed_time_250527.pdf}
    \caption{Average query processing time of \GuP, \HGMatch, \OHMiner, and \MaCH on different datasets and query sizes. The x-axis represents different query sizes, the left y-axis shows query processing time in milliseconds on a logarithmic scale, and the line graph represents the average number of embeddings (right y-axis) for each query size.}
    \label{fig:Time_consumption}
\end{figure*}

\paragraph{Datasets}
We conduct experiments on 10 real-world hypergraphs from \cite{ARBData}, which were previously used to evaluate \HGMatch \cite{HGMatch}, namely house-committees (HC), mathoverflow-answers (MA), contact-high-school (CH), contact-primary-school (CP), senate-bills (SB), house-bills (HB), walmart-trips (WT), trivago-clicks (TC), stackoverflow-answers (SA), and amazon-reviews (AR) \cite{SAdataset, TCdataset, WTdataset, ARdataset}.
We also evaluate on five unlabeled real-world hypergraphs, namely NDC-classes (NC), NDC-substances (NS), email-Eu (EU), congress-bills (CB), and coauth-MAG-history (CM) \cite{CMdataset}, by assigning the same label to all vertices.
In addition, we conduct experiments on the large hypergraphs used in \OHMiner, coauth-DBLP (CD) and AMiner (AM).
We remove all repeated hyperedges and self-loops from hypergraphs.
Additionally, for the MA and SA datasets, which contain multiple labels for vertices, we retain only the first label.
\Cref{tab:datasets} shows the statistics of the processed datasets.

\paragraph{Queries}
We generated query hypergraphs with varying sizes for each dataset, specifically containing 3, 6, 9, 12, or 15 hyperedges. For each dataset and each query size setting, we create a set of 30 query hypergraphs by using a random walk-based approach.
We do not impose any constraints on the number of vertices in the generated query hypergraphs, whereas \HGMatch limits the number of vertices to at most 35 and \OHMiner to at most 40 in their experiments \cite{HGMatch, OHMiner}.

All competitors (\HGMatch, \GuP, and \OHMiner) occasionally produce incorrect results or segmentation faults. To maintain a fair comparison, we exclude these queries from our reports on the number of solved queries, query processing time, and memory usage. As a result, the total number of queries per dataset and query size setting may be less than 30 in our experiments due to these exclusions.

\paragraph{Environment} \MaCH is implemented in C++. The source code of \HGMatch, which is implemented in Rust, was obtained from the authors.
\OHMiner and \GuP are implemented in C++ and Rust, respectively, and both are publicly available.
Experiments are conducted on a machine with two Intel Xeon Silver 4114 2.20GHz CPUs and 256GB memory.

In our analysis, we use the geometric mean when referring to average query processing times and other performance metrics. The geometric mean is particularly suitable for our experiments as these metrics vary widely across queries.
We set a time limit of 10 minutes for each query and exclude queries that none of the algorithms (\GuP, \HGMatch, \OHMiner, and \MaCH) solve within this time limit.
When an algorithm fails to solve a query that another algorithm solves, we record 10 minutes for the failing algorithm when calculating the average query processing time.

\paragraph{Preprocessing}
Following \HGMatch's methodology, we implemented a preprocessing step that builds an index of the data hypergraph prior to query processing. This index consists of signatures of all hyperedges in the data hypergraph and partitions of hyperedges based on their signatures.
Unlike \HGMatch, we do not build an inverted hyperedge index. This preprocessing is performed using only the data hypergraph, before any query hypergraphs are given.

\HGMatch's preprocessing time ranges from 0.9 ms for NC to 9,693.8 ms for AR, while our preprocessing time ranges from 1.2 ms for NC to 7260.0 ms for AM, generally increasing with the number of hyperedges in the dataset.
This preprocessing time is negligible compared to the total query processing time.
For example, on AR dataset, preprocessing takes about 7 seconds while processing 30 queries of size 15 takes more than 1.6 hours.
\OHMiner's preprocessing time ranges from 15.9 ms for HC to 139,121.0 ms for AM, and it runs out of memory for AR.
For \GuP, the preprocessing time for transforming a data hypergraph into a bipartite graph ranges from 0.3 ms for NC to 18,633.4 ms for AR.

The preprocessing time is not included in the query processing times reported in our results, as preprocessing is performed once for each dataset, rather than for each query.

\subsection{Comparison with State-of-the-Art Method}

\paragraph{Query Processing Time}
\Cref{fig:Time_consumption} presents a comparison of the average query processing time between \GuP, \HGMatch, \OHMiner, and \MaCH.
The time consumption is measured in milliseconds (ms) and displayed on a logarithmic scale. Additionally, a line graph represents the average number of embeddings for each query size.

For the EU dataset with query size 15 and the CB dataset with query size 6, \GuP, \HGMatch, and \OHMiner solve no queries within the time limit, while \MaCH solves 2 and 5 queries, respectively. For these datasets and query sizes, the average query processing times for \GuP, \HGMatch, and \OHMiner are recorded as the time limit of 10 minutes.
Likewise, in many cases including SA and AR datasets, \GuP fails to solve any queries within the time limit and its average query processing times are recorded as the time limit of 10 minutes.
For AR, \OHMiner runs out of memory (OOM) and cannot solve any queries.

The query processing time of \MaCH increases as the number of embeddings grows, demonstrating a clear correlation with the number of embeddings.
In contrast, \GuP and \HGMatch do not always follow this trend: for HC, MA, HB, and SA, they show high query processing times even when the number of embeddings is relatively low.

\MaCH outperforms its competitors in terms of query processing time, with the performance gap widening for larger query sizes.
Notably, \MaCH demonstrates particularly effective results for datasets with a large vertex label set and high average arity, i.e., HC, MA, and SA. These properties lead to more distinct hyperedge signatures, resulting in smaller candidate hyperedge sets and fewer embeddings.
The performance gap is most evident in the SA dataset for queries of size 15, where \MaCH processes queries more than 500,000 times faster than \GuP, 4,000 times faster than \HGMatch, and 14 times faster than \OHMiner.
\MaCH also shows significant performance advantages for CH with query size 15, especially compared to \OHMiner, processing queries more than 16 times faster than \GuP, 57 times faster than \HGMatch, and 80 times faster than \OHMiner.

\MaCH sometimes takes more time than \OHMiner, i.e., for small queries on HC ($q_3$) and NC ($q_3, q_6, q_9$), and for some cases on NS ($q_9, q_{12}, q_{15}$), CM ($q_3$, $q_{12}$), and CB ($q_{15}$).

As for small queries on HC and NC, these are the least time-consuming cases in our evaluation.
In these cases, our initial filtering phase can be an overhead, accounting for 89.8\% of query processing time on average.
However, as query processing time increases (due to larger query sizes or more challenging datasets), our filtering significantly reduces the number of candidates with a relatively low overhead. On AR, it spends less than 20\% of the query processing time, while achieving up to 99\% candidate reduction.

Unlabeled hypergraphs such as NS, CM and CB are particularly challenging due to the absence of vertex labels.
This absence significantly reduces the number of distinct hyperedge signatures, thus increasing the number of candidates per query.
For these unlabeled hypergraphs, MaCH may run slower than \OHMiner on some queries. However, \Cref{fig:Solved_queries} shows that \MaCH successfully solves more queries within the time limit than \OHMiner in all of these challenging datasets.
In our experiments, timeout queries were recorded as 10 minutes, which gives an advantage to algorithms that fail to finish with the time limit (10 min). For instance, on NS with query size 15, \MaCH successfully solves all queries that \OHMiner solves, plus two additional queries for which \OHMiner fails to finish. When we ran these two queries without time limits, \OHMiner takes 1,234s (20.6 mins) and 15,842s (4.4 hrs), while \MaCH takes 47s and 386s (6.4 mins), respectively. Thus, actual performances of \MaCH over \OHMiner are better than those reported in \Cref{fig:Time_consumption}.

\begin{figure}[t]
    \centering
    \includegraphics[width = 0.9\linewidth]{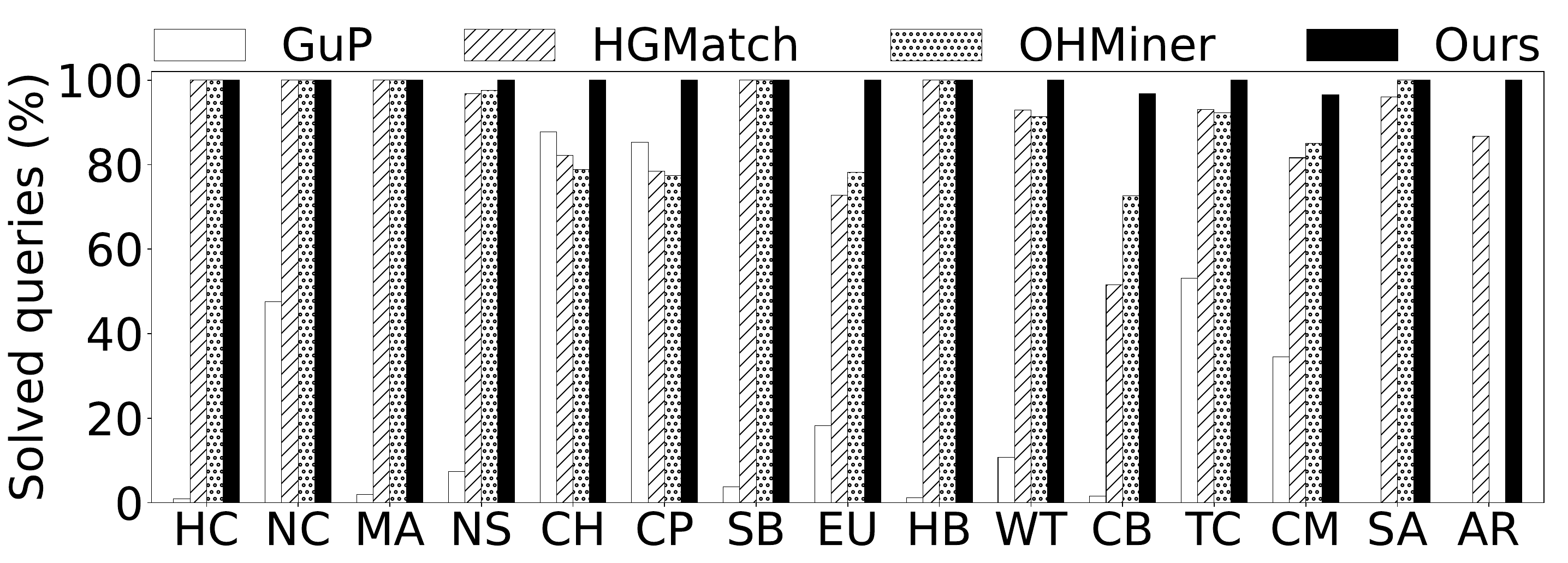}
    \caption{Ratio of solved queries within time limit of 10 minutes on different datasets.}
    \label{fig:Solved_queries}
\end{figure}

\begin{figure}[t]
    \centering
    \includegraphics[width=0.9\linewidth]{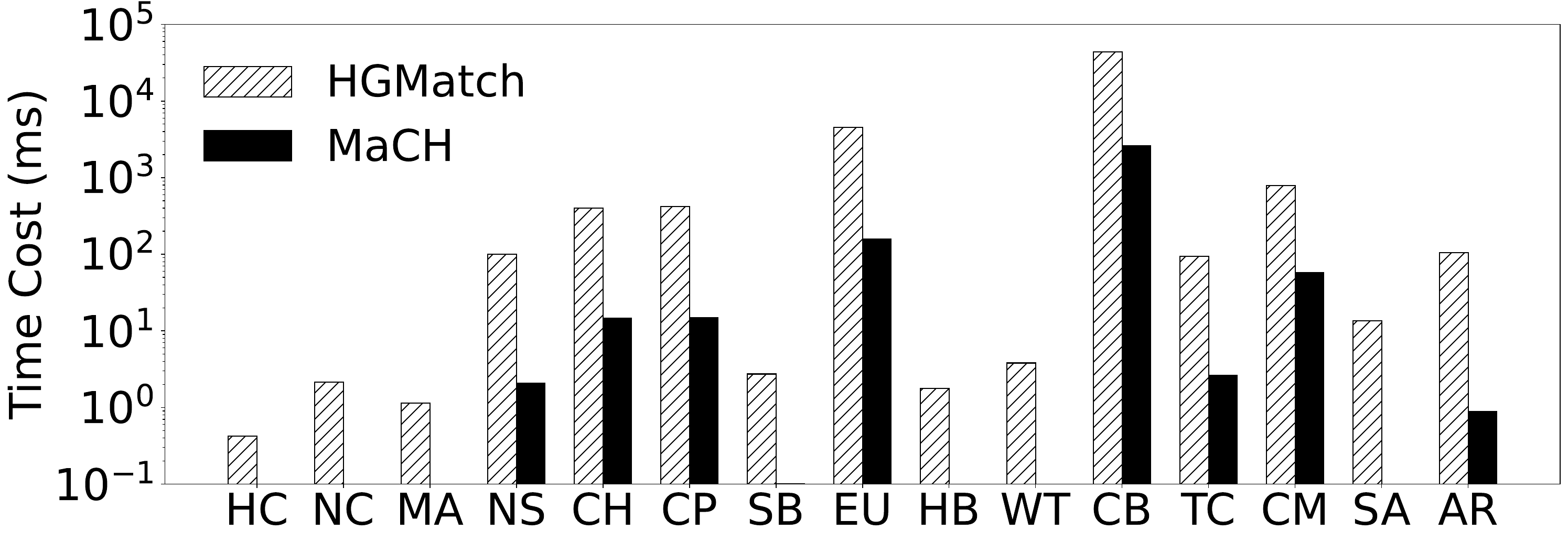}
    \caption{Average verification time of \HGMatch and \MaCH on different datasets. The y-axis represents the time in milliseconds. Bars with times less than 0.1 ms are omitted.}
    \label{fig:BreakDown}
\end{figure}

The intersection constraint is a necessary and sufficient condition for valid embeddings. Since \HGMatch also proposed a necessary and sufficient condition, equivalence of vertex profiles, we compare the two conditions for their efficiency.
\Cref{fig:BreakDown} shows the candidate verification time spent for checking the intersection constraint in \MaCH and that for checking the equivalence of vertex profiles in \HGMatch (this result includes only queries finished by both algorithms within the time limit).

For datasets whose average verification time of \MaCH is more than 0.1 ms (i.e., NS, CH, CP, EU, CB, TC, CM, and AR), \HGMatch spends 94.1\% to 96.9\% of its query processing time on the equivalence of vertex profiles, and \MaCH spends 69.4\% to 88.0\% on the intersection constraint.
These datasets require a significant amount of time to solve, highlighting the importance of an effective condition for valid embeddings. As shown in \Cref{fig:BreakDown}, our intersection constraint consistently requires less time across all datasets compared to \HGMatch's condition.
Particularly, for AR, checking the equivalence of vertex profiles in \HGMatch takes more than 100 times longer than checking the intersection constraint in \MaCH.

For remaining datasets, the number of candidates in the CHS after filtering is small, making the verification time of \MaCH negligible. As can be seen in \Cref{fig:Time_consumption}, these datasets are the least time-consuming for \MaCH.
We further analyze the effectiveness of our techniques in detail in \Cref{subsec:ablation}.

\paragraph{Ratio of Solved Queries}
\Cref{fig:Solved_queries} presents the ratio of queries solved by each algorithm to the queries solved by at least one algorithm within the 10-minute time limit.
This ratio is shown for \GuP, \HGMatch, and \MaCH across all datasets.
\MaCH successfully solves all the queries that \HGMatch, \OHMiner, or \GuP solve, except a few queries in CM and CB datasets, and outperforms them by solving additional queries in many cases.

For HC, MA, and SA, which typically have a low number of embeddings per query, \HGMatch, \OHMiner, and \MaCH generally complete the search within the time limit. \GuP, on the other hand, solves almost no queries on these datasets. \GuP tends to perform worse on datasets having high average arity (HC, MA, HB, SA, and AR) than those having low average arity (NC, CH, CP, EU, WT, TC, and CM).

CH and CP datasets also present significant challenges for \HGMatch, \OHMiner, and \MaCH.
They also have the high number of embeddings per query, caused by low maximum arity and limited number of labels.
While \GuP shows better performance than \HGMatch and \OHMiner on CH and CP, \MaCH still outperforms \GuP on these datasets.

\begin{figure}[t]
    \centering
    \includegraphics[width = 0.9\linewidth]{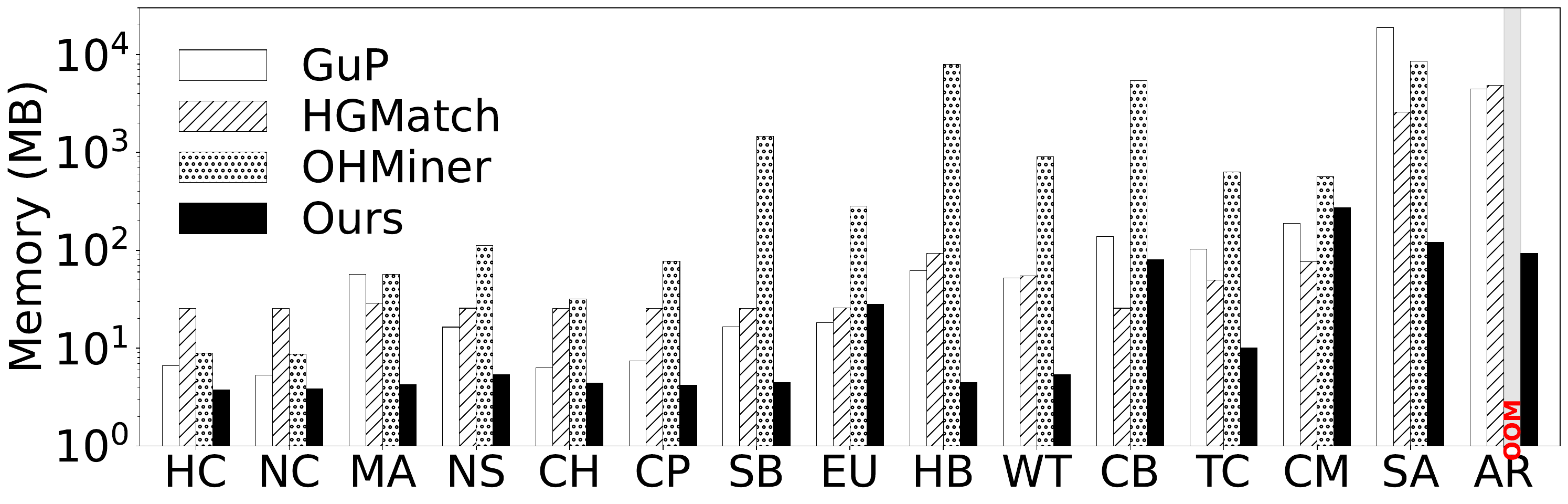}
    \caption{Peak memory consumption in megabytes, averaged over the queries for each dataset.}
    \label{fig:Memory_consumption}
\end{figure}

\paragraph{Memory Consumption}
\Cref{fig:Memory_consumption} illustrates the peak memory consumption of \GuP, \HGMatch, \OHMiner and \MaCH averaged over the queries that are solved by at least one algorithm. The memory consumption is measured in megabytes.
In general, \MaCH exhibits comparable or better memory efficiency than its competitors.
Notably, \OHMiner consumes more than 1700 times the memory of \MaCH for HB, and runs out of memory (OOM) for AR.
While memory consumption tends to increase with the size of the data hypergraph for all algorithms, \MaCH demonstrates significantly lower memory consumption on the two largest datasets (SA and AR).

\begin{figure}[t]
    \centering
    \includegraphics[width=0.9\linewidth,trim=20 0 20 0,clip]{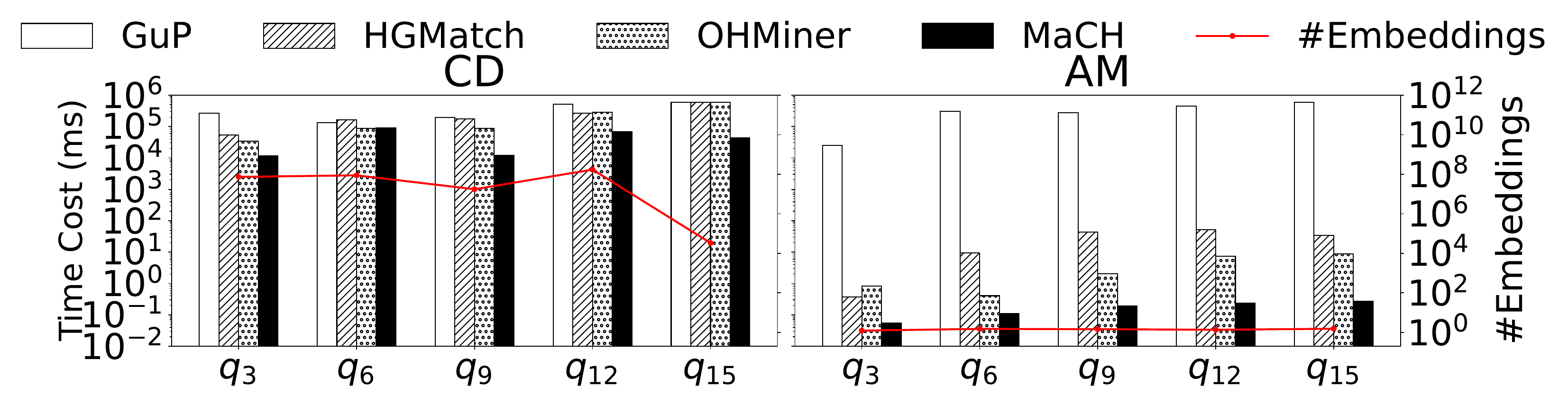}
    \caption{Average query processing time of GuP, HGMatch, OHMiner, and MaCH on CD and AM.}
    \label{fig:large_time}
\end{figure}

\paragraph{Additional Large Hypergraphs}
We conducted additional experiments on the coauth-DBLP (CD) and AMiner (AM) datasets, which are the largest hypergraphs evaluated in the \OHMiner paper.
\Cref{fig:large_time} presents a comparison of the average query processing time between GuP, HGMatch, OHMiner, and MaCH on the CD and AM datasets. On both datasets, MaCH outperforms its competitors in terms of query processing time. MaCH is particularly effective on AM where the number of vertex labels (258) is large. Specifically, MaCH is more than 1,000,000 times faster than GuP, 100 times faster than HGMatch, and 30 times faster than OHMiner for AM with query size 15. For CD, MaCH is more than 10 times faster than GuP and HGMatch, and more than 7 times faster than OHMiner on query size 9.

\paragraph{Comparison of Approaches}
The approach of transforming a hypergraph into a bipartite graph (\GuP) cannot utilize hypergraph-specific properties, leading to poor performance especially on high-arity datasets. HGMatch's on-the-fly candidate generation and repetitive checks in verification cause inefficiency even when the number of embeddings is relatively low. OHMiner generally achieves better performance than GuP and HGMatch but consumes excessive memory, making it not viable for large-scale datasets. 
In contrast, MaCH addresses these limitations through its efficient constraints, candidate hyperedge space, and Match-and-Filter framework, outperforming all competitors while maintaining comparable or better memory consumption.

\subsection{Effectiveness of Individual Techniques}
\label{subsec:ablation}

\begin{figure}[t]
    \centering
    \includegraphics[width = 0.9\linewidth]{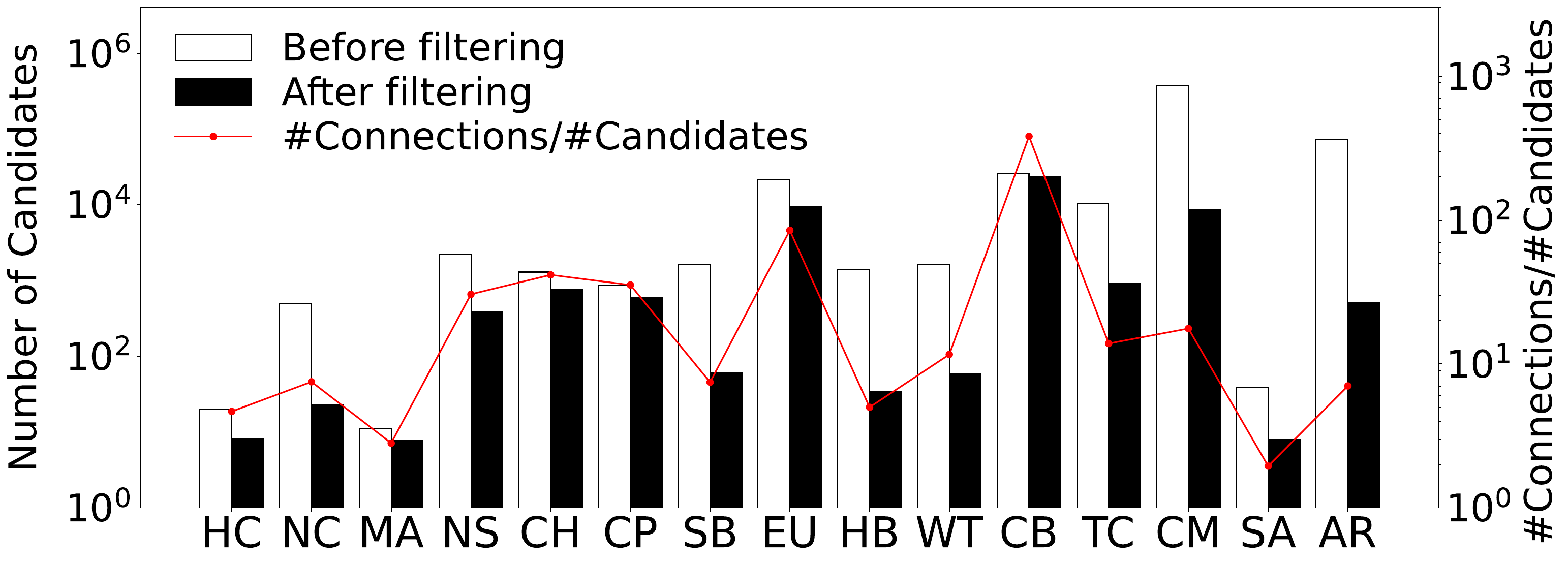}
    \caption{Average size of the CHS before and after initial filtering on different datasets. The left y-axis shows the number of candidates. A line graph represents the average number of connections per candidate (right y-axis) after initial filtering.}
    \label{fig:FilteringEffect}
\end{figure}

\paragraph{Initial Filtering by Connectivity Constraint}
\Cref{fig:FilteringEffect} illustrates the effect of our filtering technique on the candidate hyperedge space. The figure compares the average number of candidates in the CHS before and after applying our filtering method for queries.
A line graph represents the average number of connections per candidate after filtering. 

In all datasets, our filtering technique reduces the number of candidates and the number of connections significantly.
The effect of our filtering is especially evident in the AR dataset.
In AR, filtering reduces the number of candidates by 99.3\% and the number of connections by 97.7\%, eliminating a significant portion of candidates and connections from the initial CHS. It is due to the relatively large average arity and large number of labels in the dataset.

For CH and CP datasets, which have the small average arity and small number of labels (opposite to AR), our filtering still demonstrates its effectiveness.
It reduces candidates by 41.1\% and 32.1\% respectively, and connections by 47.2\% and 34.1\%.

For the CB dataset, which has the highest number of connections per candidate after filtering, the filtering process is less significant, highlighting the importance of efficient filtering during the matching process.

\begin{figure}[t]
    \centering
    \includegraphics[width =0.9\linewidth]{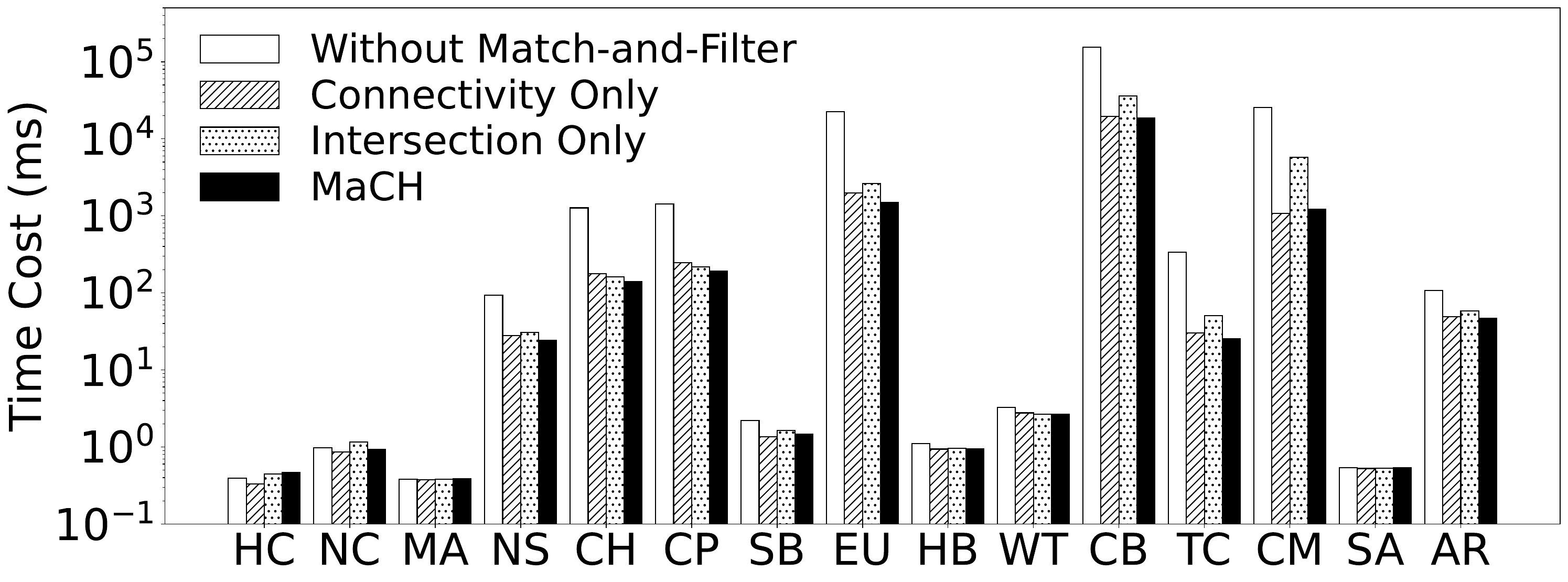}
    \caption{Average query processing time with and without filtering techniques during matching. The y-axis shows time in milliseconds.
    The first bar shows time without filtering, the second bar shows time when only filtering by intersection constraint is applied, and the third bar shows time when both intersection and connectivity constraints are applied.}
    \label{fig:CleaningTime}
\end{figure}

\paragraph{Match-and-Filter Framework}
\Cref{fig:CleaningTime} demonstrates the effect of our match-and-filter framework. The first bar shows the query processing time without match-and-filter, where no filtering techniques are used during the matching phase. In this algorithm, candidates may not be guaranteed to be $M$-compatible.
Thus, before extending a partial embedding $M$ with a candidate $(e,f)$ to $M\cup \{(e,f)\}$, we need to verify whether $(e,f)$ is $M$-compatible using the intersection constraint.
The second bar shows match-and-filter with filtering only by connectivity constraint is applied. Like the first bar, we still need to verify $M$-compatibility using the intersection constraint before extending a partial embedding.
The third bar shows filtering only by the intersection constraint is applied, thus checking $M$-compatibility as in MaCH.
The fourth bar shows MaCH, where both filtering techniques are applied.

Notably, using either constraint for filtering (second or third bar) demonstrates significant improvement over the baseline (first bar), highlighting the effectiveness of our match-and-filter framework itself.
This improvement is particularly shown in datasets with time-consuming datasets, such as EU, CB, and CM. For instance, in the CM dataset, our algorithm with both filtering techniques is 20.7 times faster.

In general, filtering only by connectivity constraint outperforms filtering only by intersection constraint because the connectivity constraint is fast ($O(1)$ time per candidate) and eliminates a considerable number of candidates in the CHS. However, it is important to note that even when filtering only by connectivity constraint is used in match-and-filter, the intersection constraint remains essential for verifying $M$-compatibility for extending partial embeddings.
Using either constraint for filtering (second or third bar) demonstrates significant improvement over the baseline (first bar), highlighting the effectiveness of our match-and-filter framework itself.
Overall, the results show that it is beneficial to use both filtering techniques together, especially for challenging datasets such as CH, CP, EU, and CB.

\section{Conclusion}
\label{sec:Conclusion}

We have proposed a novel hypergraph pattern matching algorithm that utilizes the intersection constraint, candidate hyperedge space, and Match-and-Filter framework. The intersection constraint, as a necessary and sufficient condition for valid embeddings, has potential applications beyond hypergraph pattern matching, such as hypergraph isomorphism \cite{HypergraphIsomorphism}. In addition, the combination of CHS and Match-and-Filter is a promising framework for other problems that require effective search space pruning, including subhypergraph random sampling \cite{SubhypergraphSampling} and subhypergraph cardinality estimation \cite{SubhypergraphCounting}.

\section*{Acknowledgements}
S. Song, W. Shin, and K. Park were supported by the Institute of Information \& communications Technology Planning \& Evaluation (IITP) grant funded by the Korea government (MSIT) (No. 2018-0-00551, Framework of Practical Algorithms for NP-hard Graph Problems).
G. F. Italiano is partially supported by the Italian Ministry of University and Reseach under PRIN Project n. 2022TS4Y3N- EXPAND: scalable algorithms for EXPloratory Analyses of heterogeneous and dynamic Networked Data.
W. Zhang's research was supported by ARC DP230101445 and ARC FT210100303.

\clearpage

\section*{AI-Generated Content Acknowledgement}
In this paper, we used AI models, including Claude, ChatGPT, and Gemini, to improve the readability of our writing and to assist with debugging code.

{
\balance
\bibliography{bib}
}
\unless\ifsubmit
\clearpage
\nobalance
\appendix
\section{Appendix}
\label{sec:Appendix}

\subsection{Proof of \Cref{thm:IntersectionSignature}}

{\ThmIntersectionConstraint*}

\begin{proof}
    ($\Rightarrow$) Given that $M$ is an embedding, there exists an underlying subhypergraph isomorphism, which is an injective function $\phi: V_q \to V_H$ that preserves vertex labels and satisfies $\phi(e)=M(e)$ for every $e\in E_q$. For any $S \subseteq E_q$, $\phi$ maps vertices in $\bigcap_{e\in S} e$ to those in $\bigcap_{e\in S} M(e)$, which implies $Sig(\bigcap_{e\in S} e) = Sig(\bigcap_{e\in S} M(e))$.
    
    ($\Leftarrow$) We construct a subhypergraph isomorphism $\phi$ inductively. For each subset $S \subseteq E_q$ of size $i$ (from $\abs{E_q}$ down to 1), we map every unmapped vertex in $\bigcap_{e\in S} e$ to an unmapped vertex in $\bigcap_{e\in S} M(e)$ with the matching label.

    For the base case $S = E_q$, since signatures of $\bigcap_{e\in E_q} e$ and $\bigcap_{e\in E_q} M(e)$ are the same, there are the same number of vertices for each label in $\bigcap_{e\in E_q} e$ and $\bigcap_{e\in E_q} M(e)$. Thus, there exists a bijective mapping between them that preserves labels. We define $\phi$ on the vertices in $\bigcap_{e\in E_q} e$ (as shown in dark blue in \Cref{subfig:SignatureExample}) as this bijective mapping.

    Now assume that we have constructed $\phi$ for all subsets $S$ of size greater than $k$. For a set $S$ of size $k$, consider the unmapped vertices in $\bigcap_{e\in S} e$ (as shown in blue in \Cref{subfig:CellExample}) and the unmapped vertices in $\bigcap_{e\in S} M(e)$ (as shown in blue in \Cref{subfig:CellExample_data}). Since $Sig(\bigcap_{e\in S} e) = Sig(\bigcap_{e\in S} M(e))$, there are the same number of vertices for each label in $Sig(\bigcap_{e\in S} e)$ and $Sig(\bigcap_{e\in S} M(e))$.
    For each vertex $u\in \bigcap_{e\in S} e$ which is already mapped, the corresponding vertex $\phi(u)\in \bigcap_{e\in S} M(e)$ is also mapped by the inductive assumption.
    This one-to-one correspondence between mapped vertices ensures that there are the same number of unmapped vertices for each label in $\bigcap_{e\in S} e$ and $\bigcap_{e\in S} M(e)$, and thus there exists a bijective mapping between unmapped vertices that preserves labels. We define $\phi$ on the unmapped vertices in $\bigcap_{e\in S} e$ as this bijective mapping.
    
    By induction, $\phi$ is a subhypergraph isomorphism that satisfies $\phi(e)=M(e)$ for every $e\in E_q$. That is, $M$ is a subhypergraph embedding with the underlying subhypergraph isomorphism $\phi$.
\end{proof}

\subsection{Proof of \Cref{thm:FilteringtIme}}

{\ThmFilteringtIme*}

\begin{proof}
    A pair of adjacent hyperedges $(e,e')$ is pushed to the queue initially and when candidates are removed from $C(e')$.
    That is, $(e,e')$ is pushed to the queue $O(\abs{C(e')})$ times. When $(e,e')$ is popped from the queue, it checks whether $C(e'\mid e,f)$ is empty for every $f \in C(e)$, which takes $O(\abs{C(e)})$ time in total. Therefore, the total time complexity is $O(\sum_e\sum_{e'}\abs{C(e)}\abs{C(e')})=O((\sum_e\abs{C(e)})^2)$.
\end{proof}

\subsection{Proof of \Cref{thm:HILCequivalence}}
\label{subsec:HILCequivalence}
{\ThmHILCequivalence*}
\begin{proof}
    
($\Leftarrow$)
For any label $l$, the number of vertices with label $l$ in the intersection $\bigcap_{e\in S} e$ is equal to the sum of the numbers of vertices with label $l$ across all cells that compose this intersection. Given that the signature of a cell $S$ matches the signature of $M(S)$ for every query hyperedge set $S$, it follows that the signatures of intersections, $Sig(\bigcap_{e\in S} e)$ and $Sig(\bigcap_{e\in S} M(e))$, also match for every $S$. By \Cref{thm:IntersectionSignature}, $M$ is a partial embedding.

($\Rightarrow$) Conversely, if $M$ is a partial embedding of $q'$, then by \Cref{thm:IntersectionSignature}, $Sig(\bigcap_{e\in S} e) = Sig(\bigcap_{e\in S} M(e))$ for every $S\subseteq E_{q'}$. We can derive that the signature of a cell $S$ matches the signature of $M(S)$ for every query hyperedge set $S$ using the inclusion-exclusion principle.
\end{proof}

\subsection{Proof of \Cref{thm:IScomplexity}}

{\ThmISComplexity*}
\label{subsec:ISComplexityProof}
\begin{proof}
    When we extend $M$ to $M'=M\cup\{(e,f)\}$ by mapping $e$ to $f$, each cell $S$ of $E_{q'}$ that contains vertices in $e$ is divided into two cells as follows.
    \begin{itemize}[leftmargin=*]
        \item $A = S \cap e$, i.e., the cell in $S$ which is a subset of $e$.
        \item $B = S \setminus A$, i.e., the cell in $S$ which is not a subset of $e$.
    \end{itemize}
    Since $M$ is a partial embedding, $\I_q^M(b_S,l)=\I_H^M(b_S,l)$ holds for the bitmap $b_S$ of $S$.
    Also, $\I_q^{M'}(b_A,l)=\I_H^{M'}(b_A,l)$ for the bitmap $b_A$ of the cell $A$ by the given condition.
    Since $B=S\setminus A$, 
    we have  $\I_q^{M'}(b_B,l)=\I_H^{M'}(b_B,l)$  for the bitmap $b_B$ of the cell $B$. Therefore, $\I_q^{M'}=\I_H^{M'}$ holds for every cell and label.
\end{proof}

\subsection{Memory and Ratio of Solved Queries for additional Large Hypergraphs}

\begin{figure}[t]
    \centering
    \includegraphics[width=\linewidth]{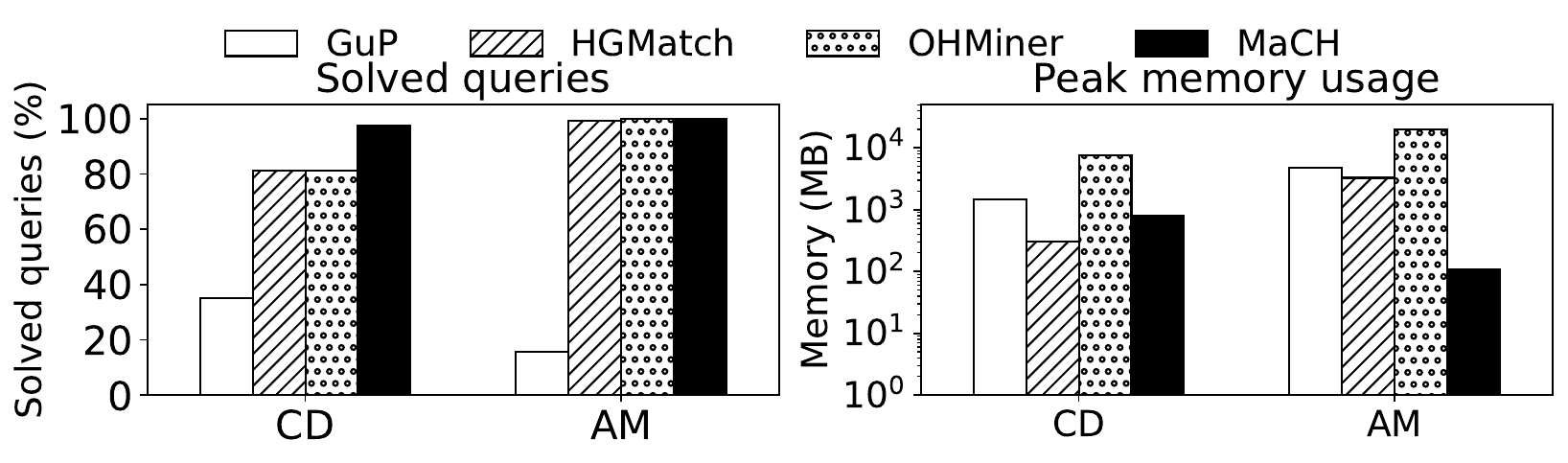}
    \caption{(Left) Ratio of solved queries within time limit of 10 minutes on CD and AM. (Right) Peak memory consumption in megabytes, averaged over the queries.}
    \label{fig:large_fig}
\end{figure}

\Cref{fig:large_fig} presents a comparison of the ratio of solved queries and the peak memory consumption between GuP, HGMatch, OHMiner, and MaCH on the CD and AM datasets.

The left figure shows the ratio of solved queries within the time limit of 10 minutes to the total number of queries solved by at least one algorithm.
MaCH consistently solved more queries compared to its competitors.

The right figure shows the peak memory consumption averaged over the queries that are solved by at least one algorithm.
MaCH exhibits comparable memory efficiency on CD and shows significantly lower memory consumption on AM, where it consumes 32 times less memory than HGMatch and 200 times less memory than OHMiner.

\subsection{Evaluation of Different Matching Orders}
\label{subsec:app_sec_matching_order}
\begin{figure}[t]
    \centering
    \includegraphics[width=0.98\linewidth]{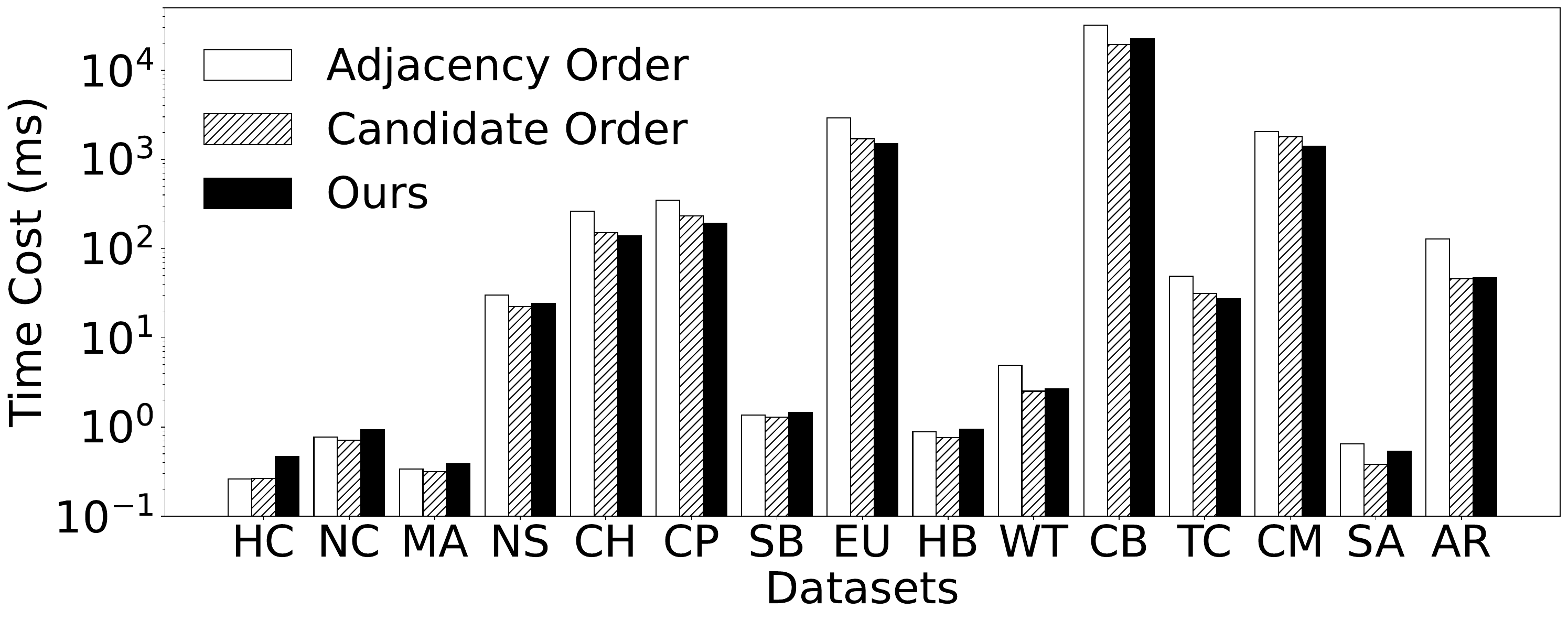}
    \caption{Average query processing time of Adjacency Order, Candidate Order, and our matching order on different datasets. The y-axis represents the time in milliseconds.}
    \label{fig:app_matching_order}
\end{figure}

In subgraph matching, two well-known heuristics are the degree heuristic and the candidate size heuristic. The degree heuristic prioritizes query vertices with the highest degree, 
and the candidate size heuristic prioritizes query vertices with the smallest number of candidates.
In hypergraphs, the degree in graphs corresponds to the number of adjacent hyperedges, and the candidate size corresponds to the number of candidate hyperedges.
Thus we consider (1) Adjacency Order, which prioritizes hyperedges with the most unmapped adjacent hyperedges, and (2) Candidate Size Order, which prioritizes hyperedges with the smallest number of candidate hyperedges.
Our matching order is a hybrid approach including the criterion that prioritizes query hyperedges with only one candidate.

\Cref{fig:app_matching_order} shows the impact of different matching order strategies on query processing time. The results show that while all three matching orders perform reasonably well, certain datasets benefit from specific strategies. Our hybrid approach achieves good performances on most of the time-consuming datasets (CH, CP, EU, and CM), validating our design choice.

\subsection{Comparative Analysis with Filtering-First Approaches}
\begin{figure}[t]
    \centering
    \includegraphics[width=0.98\linewidth]{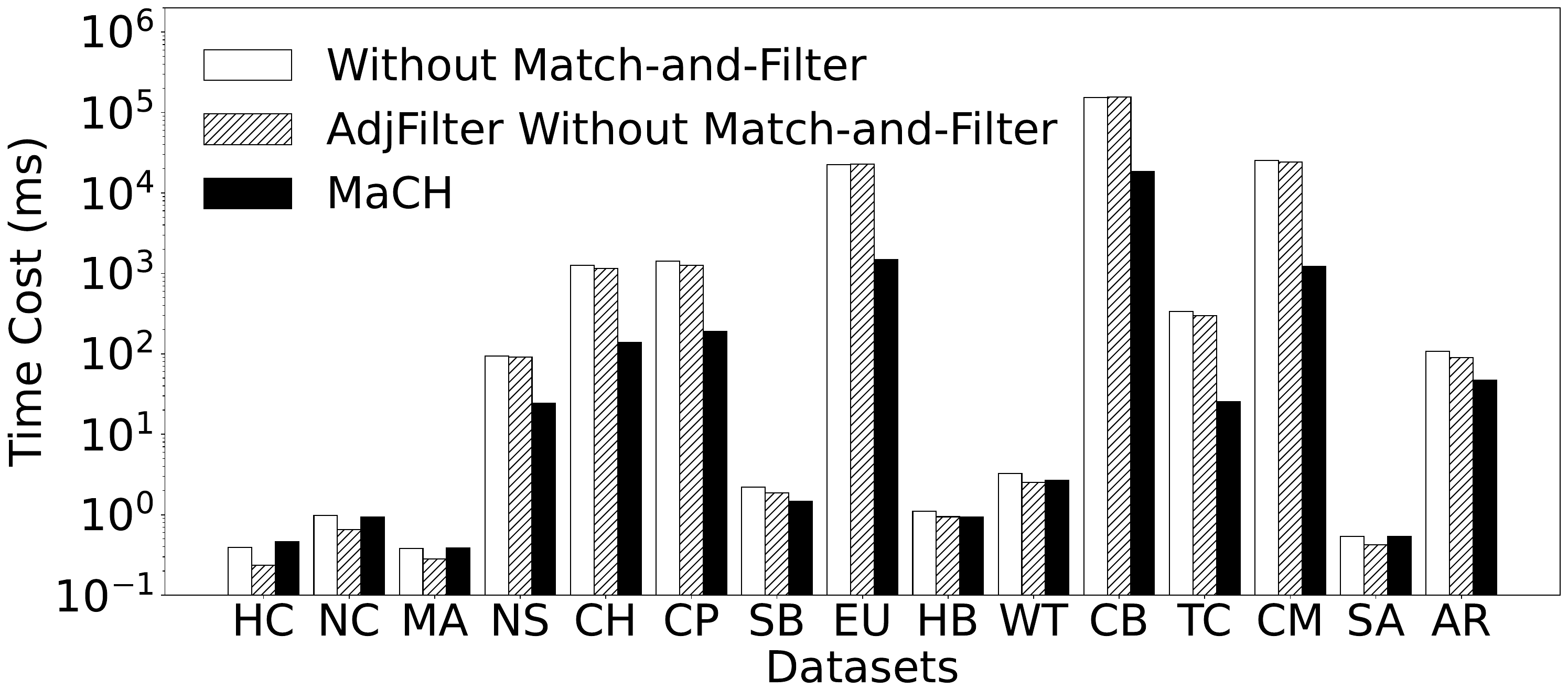}
    \caption{Average query processing time of traditional filtering-first approaches versus \MaCH. Bars show \MaCH without Match-and-Filter, Adjacency Filter without Match-and-Filter, and MaCH. The y-axis shows time in milliseconds.}
    \label{fig:AdjFiltering}
\end{figure}

For comparative analysis with traditional filtering-first approaches, we conducted experiments comparing our Match-and-Filter framework against filtering-first strategies.

Traditional filtering-first approaches, such as Ha et al. \cite{IHSfilter} apply filtering before matching begins. We implemented an enhanced version of adjacency-based filtering of Ha et al., which filters out candidates based on the number of adjacent hyperedges with matching signatures (stronger than their arity-based filter).

\Cref{fig:AdjFiltering} shows the comparison: (1) \MaCH without match-and-filter framework, (2) \MaCH without match-and-filter + adjacency filtering (AdjFilter Without Match-and-Filter), and (3) \MaCH.
Both (1) and (2) are traditional filtering-first approaches, where all filtering occurs before matching begins.
The results show the limitations of filtering-first approaches. While additional adjacency filtering provides modest improvements on small datasets (HC, NC, MA), it fails to improve performances on challenging datasets like EU, CB, and CM.

This occurs because initial filtering cannot capture which candidates will become invalid as the partial embedding grows during matching. Complex connectivity patterns involving multiple hyperedges only become apparent during the matching process. Our Match-and-Filter framework achieves 19.8 times speedup over AdjFilter Without Match-and-Filter on the CM dataset (Figure \ref{fig:AdjFiltering}).
These results validate that Match-and-Filter better handles the complex connectivity patterns inherent in hypergraphs.

\subsection{Performance Differences between Rust and C++}

\begin{figure}[t]
    \centering
    \includegraphics[width=\linewidth]{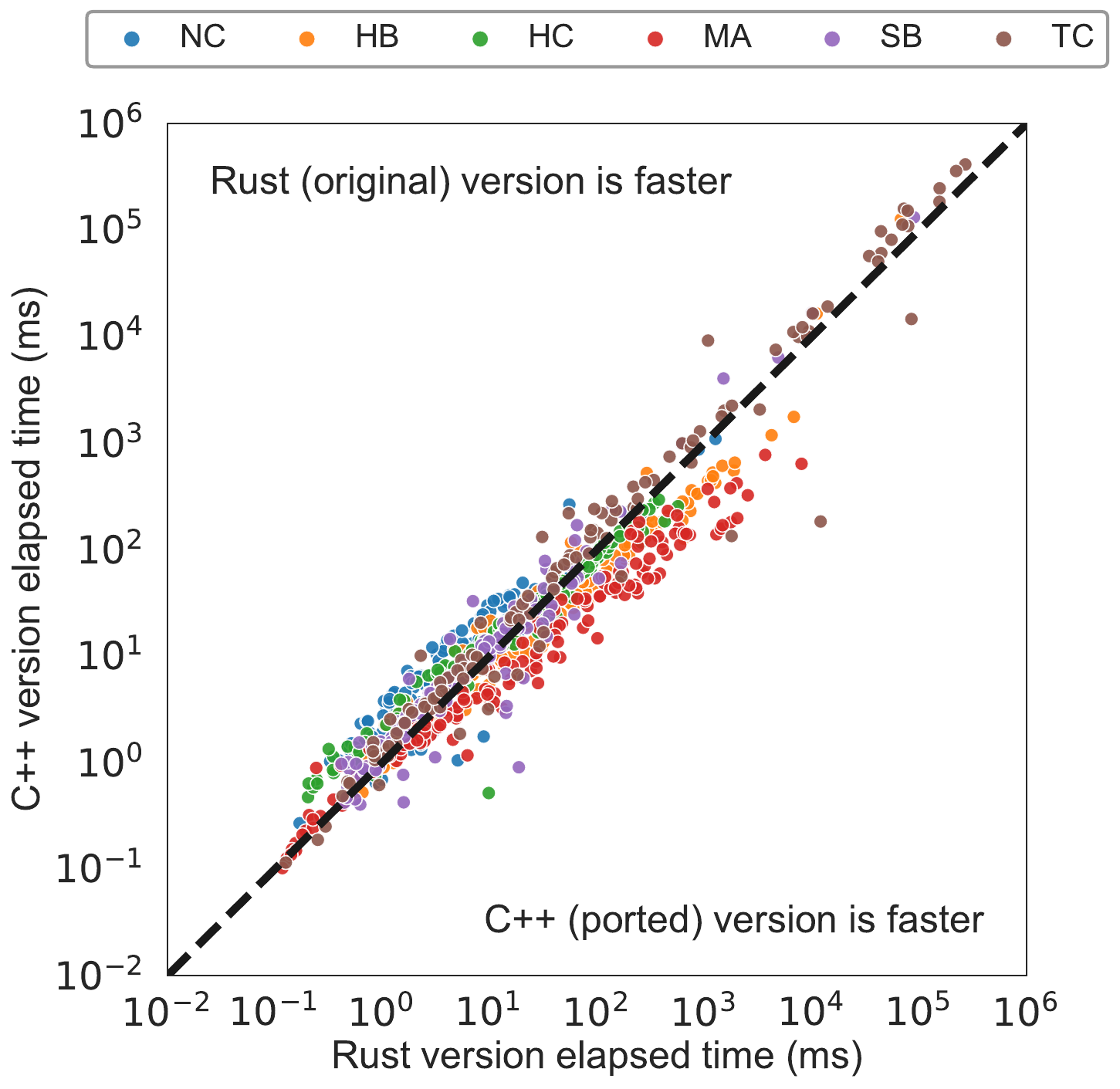}
    \caption{Comparison of query processing time for queries on HC, NC, MA, SB, HB, and TC datasets. The x-axis and y-axis represent the query processing time of Rust version (HGMatch code obtained from authors) and C++ version (ported), respectively.}
    \label{fig:lang_comp}
\end{figure}

\begin{figure}[t]
    \centering
    \includegraphics[width=\linewidth]{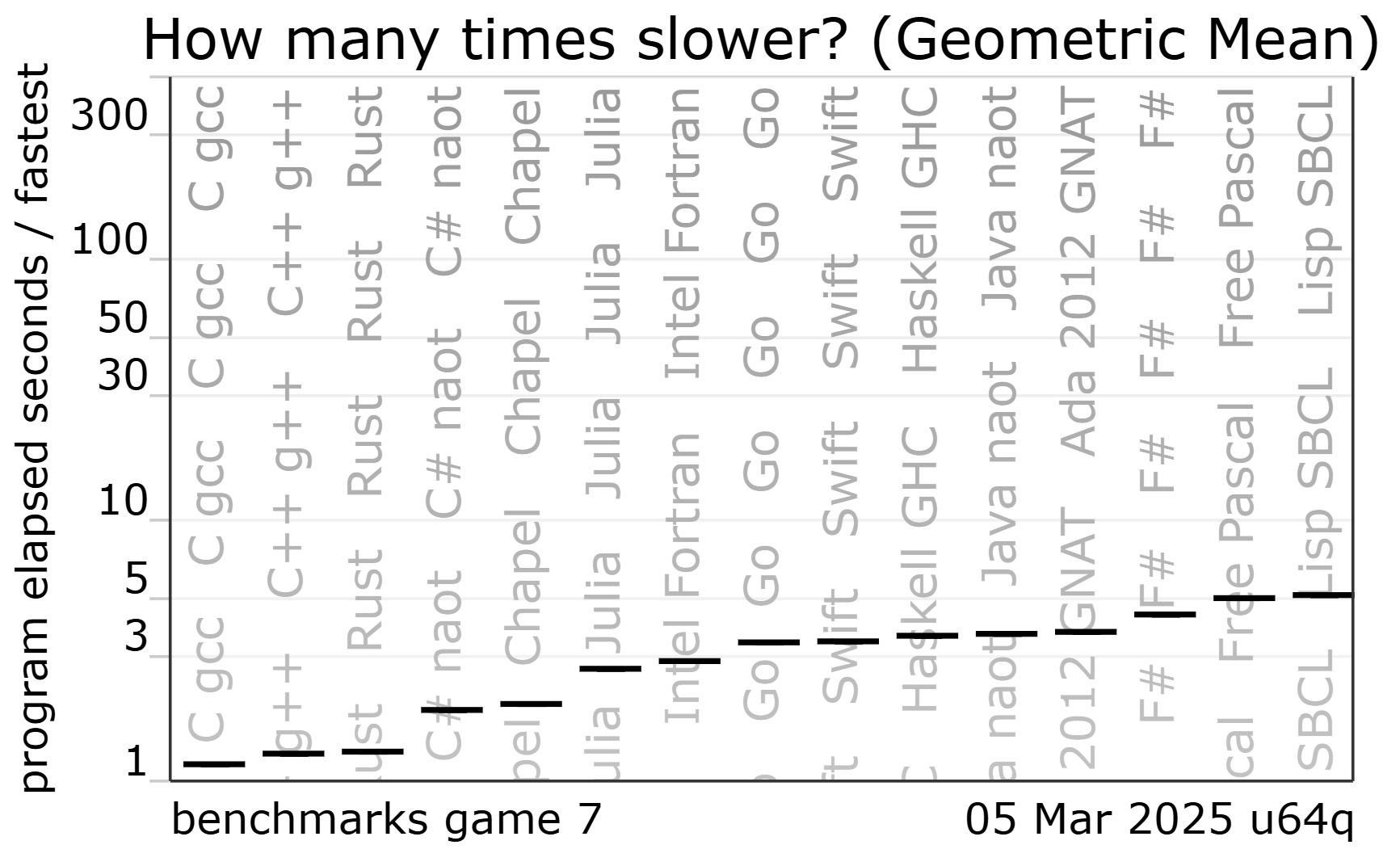}
    \caption{Running time comparison of programming languages. Numbers show how many times slower each language is on average relative to the fastest implementation for each benchmark (Lower is better). Data from Computer Language Benchmarks Game. }
    \label{fig:RustBenchmark}
\end{figure}

We converted the Rust code of \HGMatch (one of main competitors) to C++, and compared it with the original Rust version by authors. We tried to translate the code as faithfully as possible. \Cref{fig:lang_comp} shows the query processing time for queries on HC, NC, MA, SB, HB, and TC datasets. As shown in the figure, Rust and C++ versions have similar performances, with the C++ version being marginally faster (1.013 times in geometric average across all queries).

Both C++ and Rust are system programming languages with comparable performances.
According to Zhang et al.\cite{RustPerformance}, Rust programs run 1.77 times slower than C on average.
Similarly, the Computer Language Benchmarks Game\footnote{https://benchmarksgame-team.pages.debian.net/benchmarksgame/box-plot-summary-charts.html, Last accessed: Sep 1, 2025} reports how many times slower each language is on average relative to the fastest implementation for each benchmark (\Cref{fig:RustBenchmark}). The result shows that both Rust and C++ are slower than C, but they have similar performances.
Therefore, we kept the original Rust code of \HGMatch by authors in our experiments.

In addition, the performance differences we observe far exceed what could be attributed to language choice alone. Our algorithm (MaCH) and OHMiner, which shows the second-best performance in our experiments, are both implemented in C++. HGMatch and GuP are implemented in Rust. Our experiments (\Cref{fig:Time_consumption}) show that MaCH outperforms HGMatch by up to three orders of magnitude (e.g., in the SA dataset for queries of size 15) and it outperforms GuP even more than that.

\fi

\end{document}